%% file: subsetsum.tex
\documentclass[letterpaper,11pt]{article}
\usepackage{amsmath, amsthm, amssymb}
\usepackage[usenames, dvipsnames]{color}
\usepackage[normalem]{ulem} %for strikethrough
\usepackage{fullpage}
\usepackage[numbers, sort]{natbib}
\usepackage[noend]{algorithmic}
\usepackage{tikz, pgfplots}
\usepackage{hyperref}
\usepackage{xspace,color}
\usepackage{graphicx}
\usepackage{caption}
\usepackage{subcaption}
\usepackage{enumitem,linegoal}
\usepackage[linesnumbered,ruled,vlined]{algorithm2e}
\usepackage{comment}	
\usepackage{ctable}

\usepackage{mathtools}
\usepackage[flushleft]{threeparttable}
\usepackage{verbatim}
\usepackage{subcaption}
\usepackage{thmtools}
\usepackage{thm-restate}
\usepackage{cleveref}
\usepackage{fullpage}
\usepackage{graphicx}
\usepackage{comment}
\usepackage{mathtools}
\usepackage[linesnumbered,ruled,vlined]{algorithm2e}
\usepackage[utf8]{inputenc}
\usepackage{tikz}
\usepackage{amsthm}
\usepackage{subcaption}
\usepackage{amsfonts}
\usepackage{hyperref}
\usepackage[framemethod=TikZ]{mdframed}
\usepackage{amsthm}
\usepackage{thmtools}
\usepackage{thm-restate}
\usepackage[normalem]{ulem}
\usepackage{filecontents}

\usepackage{bbm}
\usepackage{thm-restate}

\usepackage{footnote}
\makesavenoteenv{tabular}
\makesavenoteenv{table}

\usepackage[margin=1in]{geometry}

\definecolor{crimsonglory}{rgb}{0,0,0}%{0.75, 0.0, 0.2}

 \newtheorem{theorem}{Theorem}[section]
 
 \newtheorem{lemma}[theorem]{Lemma}

 \newtheorem{remark}[theorem]{Remark}
\makeatletter
\def\GrabProofArgument[#1]{ #1: \egroup\ignorespaces}
\def\proof{\noindent\textbf\bgroup Proof%
	\@ifnextchar[{\GrabProofArgument}{. \egroup\ignorespaces}}

\makeatother

\newenvironment{theorem-repeat}[1]{\begin{trivlist}
		\item[\hspace{\labelsep}{\bf\noindent Theorem \ref{#1} (repeated).}]\em }%
	{\end{trivlist}}

\usetikzlibrary{arrows,shapes,snakes,automata,backgrounds,petri,calc}

\input{src/macros.tex}

\newcounter{proccnt}

\newcommand{\konote}[1]{}

\usepackage[normalem]{ulem} %for strikethrough

\title{Dynamic Subset Sum with Truly Sublinear Processing Time}
\usepackage{authblk}

\author[1,2]{Hamed Saleh}
\author[2]{Saeed Seddighin}
\affil[1]{University of Maryland, College Park}
\affil[2]{Toyota Technological Institute at Chicago}

\begin{document}
	\newcommand{\ignore}[1]{}
\renewcommand{\theenumi}{(\roman{enumi}).}
\renewcommand{\labelenumi}{\theenumi}
\sloppy
\date{}
\author{}

\maketitle

\thispagestyle{empty}
\begin{abstract}
\input{src/abstract}
\end{abstract}

\input{src/introduction}

\input{src/preliminaries}
\input{src/results}

\input{src/main}

\input{src/offline}
\input{src/remove-hardness}
\bibliographystyle{abbrv}	
\bibliography{src/references}
\appendix
\newpage
\input{src/3sum}

\input{src/ksum}

\input{src/3sum-hard}

\end{document}

%% file: src/macros.tex
\newcommand{\tmax}{\mathsf{t}_{\mathsf{max}}}
\newcommand{\rmax}{\mathsf{r}_{\mathsf{max}}}

\newcommand{\opers}{\mathsf{oprs}}
\newcommand{\important}{\mathcal{I}}
\newcommand{\inst}{\xi}
\newcommand{\sums}{\mathcal{C}}
\newcommand{\milestones}{\lambda}

\newcommand{\Aa}{\mathsf{A}}
\newcommand{\Bb}{\mathsf{B}}
\newcommand{\Cc}{\mathsf{C}}
\newcommand{\Xx}{\mathsf{X}}
\newcommand{\Yy}{\mathsf{Y}}
\newcommand{\Zz}{\mathsf{Z}}

\newcommand{\Saeed}[1]{{\color{red}[Saeed: #1]}}
\newcommand{\Hamed}[1]{{\color{blue}[Hamed: #1]}}

%% file: src/abstract.tex
Subset sum is a very old and fundamental problem in theoretical computer science. In this problem, $n$ items with weights $w_1, w_2, w_3, \ldots, w_n$ are given as input and the goal is to find out if there is a subset of them whose  weights sum up to a given value $t$. While the problem is NP-hard in general, when the values are non-negative integer, subset sum can be solved in pseudo-polynomial time $~\widetilde O(n+t)$. 

In this work, we consider the dynamic variant of subset sum. In this setting, an upper bound $\tmax$ is provided in advance to the algorithm and in each operation, either a new item is added to the problem or for a given integer value $t \leq \tmax$, the algorithm is required to output whether there is a subset of items whose sum of weights is equal to $t$. Unfortunately, none of the existing subset sum algorithms is able to process these operations in truly sublinear time\footnote{Truly sublinear means $n^{1-\Omega(1)}$.} in terms of $\tmax$.

Our main contribution is an algorithm whose amortized processing time\footnote{Since the runtimes are amortized, we do not use separate terms update time and query time for different operations and use processing time for all types of operations.} for each operation is truly sublinear in $\tmax$ when the number of operations is at least $\tmax^{2/3+\Omega(1)}$. We also show that when both element addition and element removal are allowed, there is no algorithm that can process each operation in time $\tmax^{1-\Omega(1)}$ on average unless \textsf{SETH}\footnote{The \textit{strong exponential time hypothesis} states that no algorithm can solve the satisfiability problem in time $2^{n(1-\Omega(1))}$.} fails. 

%% file: src/introduction.tex
\section{Introduction}
Dynamic programming (DP\footnote{To avoid confusion with the dynamic setting, we refer to dynamic programming by DP throughout the paper.}) is one of the basic algorithmic paradigms for solving computational problems. 
The technique is of central use in various
fields such as bioinformatics, economics, and operations research.
While DP has been a core algorithmic technique in the
past, as computational models evolve, there is a pressing need to design alternative algorithms for DP problems.

A great example is when the algorithm is required to update the solution after each incremental change, namely the dynamic setting. Dynamic settings for many problems have been studied, e.g. \cite{DBLP:conf/stoc/HenzingerKNS15,DBLP:conf/stoc/NanongkaiS17,gawrychowski2018optimal,DBLP:conf/stoc/AssadiOSS18,DBLP:conf/soda/AssadiOSS19,chen2013dynamic,DBLP:conf/stoc/LackiOPSZ15,DBLP:journals/corr/abs-1909-03478,DBLP:conf/focs/NanongkaiSW17}. In general, in dynamic settings the goal is to develop an algorithm where the solution can be updated efficiently given incremental changes to the input. In the context of graph algorithms~\cite{DBLP:conf/stoc/NanongkaiS17,DBLP:conf/focs/NanongkaiSW17,DBLP:conf/stoc/LackiOPSZ15,DBLP:conf/stoc/AssadiOSS18,DBLP:conf/soda/AssadiOSS19,DBLP:journals/corr/abs-1909-03478}, such changes are usually modeled by edge addition or edge deletion. Also, for string problems, changes are typically modeled with character insertion and character deletion~\cite{charalampopoulos2020dynamic,saeednew,our-stoc-paper,gawrychowski2020fully}.

This work is concerned with dynamic algorithms of subset sum which is a text-book example of DP. In this problem, a list of $n$ items with weights $w_1, w_2, \ldots, w_n$ are given as input and the goal is to find out if there is a subset of items whose weights sum up to $t$. While classic DP algorithms solve subset sum in time $O(nt)$, more efficient algorithms have been recently developed for subset sum~\cite{koiliaris2019faster}, culminating in the $\tilde O(n+t)$ time algorithm of Bringmann~\cite{bringmann2017near}.

We consider the dynamic variant of subset sum wherein an upper bound $\tmax$ is given to the algorithm in advance and in each operation, either a new item is added to the problem or for a given value $t \leq \tmax$, the algorithm is required to output whether there is a subset of items whose sum of weights is equal to $t$. The classic DP solution for subset sum can be modified to process each operation in time $ O(\tmax)$. Unfortunately, even the more advanced algorithms of subset sum~\cite{koiliaris2019faster,bringmann2017near} do not offer a meaningful improvement upon this naive algorithm (the only way one can make use of them in the dynamic context is to recompute the solution from scratch after each modification). Therefore, a natural question that emerges is if one can design a dynamic algorithm for subset sum that processes each operation in truly sublinear time in terms of $\tmax$?

In this work, we answer this question in the affirmative. Our main contribution is an algorithm whose amortized processing time for each operation is truly sublinear in $\tmax$ when the number of operations is at least $\tmax^{2/3+\Omega(1)}$. We also show that when both element addition and element removal are allowed, there is no algorithm that can process each operation in truly sublinear time unless \textsf{SETH} fails. In addition to this, we present several algorithms for simplified variants of subset sum that outperform our general algorithm.
\setcounter{page}{1}

\input{tables/results}

\subsection{Related Work}
Recent studies have investigated alternative dynamic algorithms for DP problems. For instance, a series of recent works~\cite{saeednew,our-stoc-paper,gawrychowski2020fully} (STOC'20,21,21) give dynamic algorithms for the \textit{longest increasing subsequence problem}  that update the solution in sublinear time after each modification to the input. Other famous DP problems such as \textit{edit distance} and \textit{the longest common subsequence problem} have also been studied in the dynamic setting and almost efficient algorithms are given for them~\cite{charalampopoulos2020dynamic}. 

The well-known pseudo-polynomial time algorithm of Bellman~\cite{bellman57a} solves the subset sum problem in time $O(nt)$. No significant improvements upon this classic algorithm was known until recent years. In 2017, Koiliaris and Xu~\cite{koiliaris2019faster} break the $O(nt)$ barrier after several decades by giving two algorithms with runtimes $O(n\sqrt{t})$ and $O(n + t^{4/3})$. Shortly after, Bringmann~\cite{bringmann2017near} presents an algorithm which solves the subset sum problem in time $\widetilde{O}(n+t)$. This algorithm is tight up to polylogarithmic factors.

When $t$ is exponentially large in terms of $n$, subset sum can be solved in time $O(2^{n/2} \textsf{poly}(n, \log t))$ using the meet-in-the-middle approach~\cite{horowitz1974computing}. On the flip side, Abboud, Bringmann, Hermelin, and Shabtay~\cite{abboud2019seth} show that for any constant $\epsilon > 0$, there exists no algorithm with runtime $t^{1-\epsilon}2^{o(n)}$ for subset sum unless \textsf{SETH} fails. Their result is based on a careful reduction from $k$-\textsf{SAT} to subset sum.

Another variant of the subset sum problem which has received attention in recent years is the \textit{modular subset sum} problem. This problem is similar to subset sum except that an additional parameter $m$ is provided and the goal is to find out if there exists a subset of items such that the total sum of their weights modulo $m$ is equal to $t$. This problem has been studied in a series of works~\cite{koiliaris2019faster,axiotis2019fast,axiotis2021fast,cardinal2021modular} and it is known to be solvable in time $\tilde O(n+m)$.

%% file: tables/results.tex
\definecolor{LightCyan}{rgb}{0.88,1,1}
\begin{table}[h]
    \renewcommand{\arraystretch}{1.2}
	\centering
	\begin{tabular}{|c|c|c|c|}
		\hline
		problem & amortized processing time & reference\\
		\hline
		dynamic subset sum & $\tilde O(1 + \tmax^{5/3} /\opers)$ & Theorem~\ref{theorem:main}\\
		\hline	
        dynamic subset sum with offline access & $\tilde O(1+\tmax \min{\{\sqrt{\opers},\sqrt{\tmax}\}}/\opers)$ & Theorem~\ref{thm:offline}\\
        \hline 
        fully dynamic bounded 3-sum & $\tilde O(\min{\{\opers, \rmax^{0.5}\}})$ & Theorem~\ref{thm:3sum}\\
        \hline 
        fully dynamic bounded $k$-sum & $\tilde O(\min{\{\opers^{k-2}, \rmax^{\frac{k-2}{k-1}}\}})$ & Theorem~\ref{thm:ksum}\\
        \hline
	\end{tabular}
\caption{The results of this paper are summarized in this table. $\opers$ denotes the number of operations in the dynamic problems.}
\end{table}

%% file: src/preliminaries.tex
\subsection{Preliminaries}
In the subset sum problem, we are given a list of $n$ items with non-negative integer weights $w_1, w_2, \ldots, w_n$, and we are asked to find out whether there exists a subset $S \subseteq [n]$ such that $\sum_{i \in S}{w_i} = t$ for a given parameter $t$. We define the dynamic variant of the subset sum problem in the following way:

\begin{itemize}
	\item There is a multiset of items which is initially empty. Also, a value $\tmax$ is given to the algorithm before any of the operations arrives.
	\item At each point in time, one of the following operations arrives:
	\begin{enumerate}
		\item Either an item with a given weight $w$ is added to the problem. We call such operations \textit{insertion operations}.
		\item Or for a given value $0 \leq t \leq \tmax$, the algorithm is required to find out if there is a subset of items whose total sum of weights is equal to $t$. We call such operations \textit{query operations}.
	\end{enumerate}
\end{itemize}

We also consider another variant of dynamic subset sum which we call \textit{fully dynamic subset sum}. In the fully dynamic regime, in addition to insertion operations and query operations, our algorithm is required to process deletion operations in which an already existing item is deleted from the problem. 

Another setting that we consider for dynamic subset sum is when offline access is provided. We call this problem dynamic/fully dynamic subset sum with offline access. This setting is similar to the dynamic setting, except that all operations are given to the algorithm at once and only after that the algorithm is required to process the operations in their order. Indeed, this setting is easier than the dynamic setting since the algorithm can foresee future operations before processing them.

Our goal is to design an algorithm that can process the operations in truly sublinear amortized time in terms of $\tmax$ and indeed this is only possible if the number of operations is considerable. Therefore, we also denote the number of operations by $\opers$. It is important to note that $\opers$ is \textbf{not} known to the algorithm in advance (unless offline access is provided).

We define the bounded 3-sum problem as a simplified variant of subset sum in the following way: In the bounded 3-sum problem, three sets $\Aa$, $\Bb$, and $\Cc$ are given as input wherein each set contains integer numbers in range $[0, \rmax]$. The goal is to find out if there are three numbers $a \in \Aa$, $b \in \Bb$, and $c \in \Cc$ such that $a+b = c$. Similar to subset sum, we define the dynamic variant of bounded 3-sum in the following way:

\begin{itemize}
	\item Initially, sets $\Aa$, $\Bb$, and $\Cc$ are empty. Also, the value of $\rmax$ is given to the algorithm before any of the operations comes.
	\item At each point in time, one of the following operations arrives:
	\begin{enumerate}
		\item Either an integer number $w$ in range $[0, \rmax]$ will be added to $\Aa$, $\Bb$, or $\Cc$. We call such operations \textit{insertion operations}.
		\item Or the algorithm is asked if there are three numbers $a \in \Aa$, $b \in \Bb$, and $c \in \Cc$ such that $a+b = c$. We call such operations \textit{query operations}.
	\end{enumerate}
\end{itemize}

When deletion operations are also allowed, we call the problem fully dynamic bounded 3-sum. We also generalize bounded 3-sum to bounded $k$-sum wherein instead of three sets $\Aa$, $\Bb$, and $\Cc$, we have $k$ sets $\Aa_1$, $\Aa_2$, $\ldots$, $\Aa_k$ and the goal is to find out if there are numbers $a_1 \in \Aa_1, a_2 \in \Aa_2, \ldots, a_k \in \Aa_k$ such that $a_1 + a_2 + a_3 + \ldots, + a_{k-1} = a_k$. Similar to dynamic bounded 3-sum, in the dynamic version of bounded $k$-sum, the value of $\rmax$ is given to the algorithm in advance. Insertion operations can add any integer number in range $[0, \rmax]$ to any of the sets $\Aa_1, \Aa_2, \ldots, \Aa_k$. Also, in the case of fully dynamic bounded $k$-sum, each deletion operation will remove an already existing number from a set. For a query operation, the algorithm has to find out if there are numbers $a_1 \in \Aa_1, a_2 \in \Aa_2, \ldots, a_k \in \Aa_k$ such that $a_1 + a_2 + a_3 + \ldots, + a_{k-1} = a_k$.

Throughout the paper, when we say a runtime is truly sublinear, we mean the runtime is truly sublinear in terms of $\tmax$ for subset sum, or truly sublinear in terms of $\rmax$ in case of $3$-sum and $k$-sum, unless explicitly stated otherwise.

%% file: src/results.tex
\section{Our Results and Techniques}
Our main contribution is an algorithm for dynamic subset sum whose amortized processing time is truly sublinear in terms of $\tmax$ when the number of operations is at least $\tmax^{2/3+\Omega(1)}$. 

\vspace{0.2cm}
{\noindent \textbf{Theorem}~\ref{theorem:main}, [restated informally]. \textit{There exists an algorithm for dynamic subset sum whose amortized processing time is bounded by $\tilde O(1 + \tmax^{5/3} /\opers)$. The algorithm gives a correct solution to each query operation with high probability.\\}}

\begin{figure}[ht]
	\centering
	\input{figs/online-time}
	\caption{The solid polyline shows the amortized processing time of our dynamic subset sum algorithm. The red part indicates the range in which the trivial $\tilde{O}(\tmax)$ algorithm is optimum, and the orange part indicates the range in which $\tilde{O}(1+\tmax^{5/3}/\opers)$ is optimum. The dashed blue polyline illustrates the lower bound that is implied from~\cite{abboud2019seth}.}\label{fig:online-time}
\end{figure}
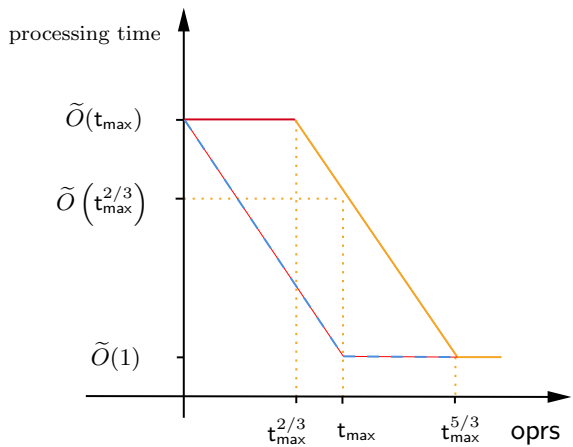

Let us clarify a few points here. First and foremost, our techniques are new and novel and to the best of our knowledge have not been used in previous work\footnote{Unfortunately, for the previous submission of this paper, a reviewer questioned the novelty of our techniques without brining any similar results.}. Second, in our setting,  multiple items can have equal weights. However, the problem is non-trivial even if we are guaranteed that all item weights are distinct and indeed our algorithm still obtains a sublinear processing time when the number of operations is at least $\tmax^{2/3+\Omega(1)}$\footnote{Unfortunately, for the previous submission of this paper, a reviewer saw this feature of our algorithm as a negative.}. As we explain below, our dynamic algorithm directly implies a new subquadratic time algorithm for subset sum in the standard setting (which is a fairly recent result) and therefore it cannot be implied from old works on subset sum\footnote{Unfortunately, for the previous submission of this paper, a reviewer was not convinced that our result cannot be implied from works in 1990s on subset sum. We hope this addresses the issue.}.

Theorem~\ref{theorem:main} does not make use of existing subquadratic time algorithms for subset sum and thus it can be seen as an independent improved algorithm for subset sum which is able to solve the problem iteratively. In other words, the result of Theorem~\ref{theorem:main} leads to an algorithm for subset sum which runs in time $\tilde O(n+t^{5/3})$. We remark that although the runtime is subquadratic in terms of $n+t$, this algorithm is still slower than the more advanced algorithms of subset sum~\cite{koiliaris2019faster, bringmann2017near}.

Before explaining the techniques, we would like to justify some of the limitations of our algorithm. Since our guarantees do not carry over to the fully dynamic setting wherein deletion operations are also allowed, it is natural to ask whether a fully dynamic algorithm for subset sum with truly sublinear processing time is possible at all?  We prove that such an algorithm does not exist unless \textsf{SETH} fails.

\vspace{0.2cm}
{\noindent \textbf{Theorem}~\ref{thm:hardness}, [restated informally]. \textit{Assuming \textsf{SETH}, there does not exist a fully dynamic subset sum algorithm that processes the operations in amortized truly sublinear time even if the number of operations grows to $\tmax^c$ for some constant $c > 0$.\\}}

It is interesting to note that Theorem~\ref{thm:hardness} holds even if offline access to the operations is provided in advance. Theorem~\ref{thm:hardness} is mainly inspired by the  \textsf{SETH}-based hardness of~\cite{abboud2019seth} for subset sum. In our proof, we build a collection of $N = \Theta(\tmax^c)$ instances of the subset sum problem $\inst_1, \inst_2, \ldots, \inst_N$ each with $O(\log \tmax)$ items and target $t_i$ which is bounded by $\tmax$. 
We denote by $\inst_i^+$ a set of operations that adds the items of $\inst_i$ to the dynamic problem, and by $\inst_i^-$ a set of operations that removes the items of $\inst_i$ from the dynamic problem. In addition, let $\inst_i^\textsf{?}$ be a query operation that asks if we can build a sum of $t_i$ using the elements of the dynamic problem at the time of the query operation. Then we consider the following sequence of operations for the dynamic subset sum problem which has a combination of insertion, deletion, and query operations:

\[ \inst_1^+, \inst_1^\textsf{?}, \inst_1^-, \inst_2^+, \inst_2^\textsf{?}, \inst_2^-, \ldots, \inst_N^+, \inst_N^\textsf{?}, \inst_N^- \] 

Roughly speaking, we show that under the  \textsf{SETH}  hypothesis, there is no algorithm that is able to find out whether the result of any of $\inst_i$'s is Yes in time $O(N \tmax^{1-\epsilon})$  for any constant $\epsilon > 0$. 
%Note that we set $\gamma = 1$ for simplicity. 
However if we feed this sequence into a fully dynamic subset sum algorithm with amortized processing time $O(\tmax^{1-\epsilon})$, it will solve this problem in time $\tilde O(N \tmax^{1-\epsilon})$ which contradicts \textsf{SETH}. We explain this reduction in details in Section~\ref{sec:hardness}.

Another limitation of our algorithm is that although its processing time is truly sublinear on average, for a few operations the runtime may grow to linear. This raises the question of whether it is possible to design a dynamic algorithm for subset sum whose worst-case runtime is truly sublinear? The answer to this question is also negative since~\cite{abboud2019seth} implies that solving subset sum for an instance with $n = O(\log \tmax)$ items and target $t= \tmax/2$ cannot be done in time $O(\tmax^{1-\epsilon})$ for any constant $\epsilon > 0$ assuming \textsf{SETH}. Let $\mathsf{I}$ be such an instance of subset sum. Consider a dynamic scenario for subset sum in which for an arbitrary number of operations, items with values more than $\tmax/2$ are added to the problem and at some point we decide to add the items of instance $\mathsf{I}$ to the dynamic problem. After we add all those items, we ask if there is a subset of items whose total weight is equal to $\tmax/2$. The hardness of~\cite{abboud2019seth} proves that the total time our algorithm spends on the last $n+1$ operations cannot be truly sublinear unless \textsf{SETH} fails. Moreover, since $n$ is polylogarithmic in terms of $\tmax$, the runtime of processing at least one of those operations would not be truly sublinear.

The above argument also shows that in order to obtain a dynamic algorithm whose amortized processing time is truly sublinear, there has to be a lower bound on the number of operations. Otherwise, the hardness of~\cite{abboud2019seth} can be used to make a dynamic subset sum scenario in which the number of operations is polylogarithmic and that there is no hope to obtain a truly sublinear processing time on average unless \textsf{SETH} fails. We acknowledge that this argument does not prove a lower bound of $\tmax^{2/3+\Omega(1)}$ on the number of operations which is required by our algorithm in order to make sure the amortized processing time is truly sublinear.

\subsection{Dynamic Subset Sum with Truly Sublinear Processing Time}
We are now ready to state the main ideas of our algorithm. We begin by considering a more relaxed version of the subset sum problem which is a great illustrating example of the ideas behind our algorithm. We call this problem \textit{static subset sum}. In this problem, an initial array $A: [0,\tmax] \rightarrow \{0,1\}$ is given which represents a solution to the subset sum problem. That is, $A_i = 1$ holds if and only if there is a subset of items whose total sum of weights is equal to $i$. The goal is to preprocess the sequence such that we can find out how the solution array changes if we add an element with a given weight  to the problem. More precisely, for a given $w$, we would like to find indices $w \leq i \leq \tmax$ such that $A_i = 0$ and $A_{i-w} = 1$. However, the algorithm does not update the sequence after an operation is processed. In other words, the initial sequence remains intact throughout the algorithm and we only need to detect which places of array $A$ would change if we added a new element. After an operation is processed, we do not modify the initial sequence.

We begin by stating a well-known application of \textit{Fast Fourier Transform} that is useful in this context. Given $0/1$ arrays $X$ and $Y$, one can construct in time $O((|X|+|Y|) \log (|X| + |Y|))$ an array $Z$ of size $|X|+|Y|+1$, such that $Z_i = \sum X_j Y_{i-j}$. This technique is usually referred to as \textit{polynomial multiplication}. Now, let $X$ be exactly the same as array $A$ and $Y$ be made as follows: $Y_i = 1-A_{|A|-i+1}$. (Both $X$ and $Y$ have the same length as $A$.) It follows that if we compute array $Z$ using polynomial multiplication, the number of indices of $A$ that would be affected by adding an item with any weight $w$ will be available in $Z$. More precisely, $Z_{|Y|+1 - w}$ would be equal to the number of indices $i$ such that $A_{i} = 0$ but $A_{i-w} = 1$. This technique would solve the problem if we were to just report the number of affected indices of $A$ after each operation instead of reporting their locations. In that case, the preprocessing time of the algorithm would be $O(|A| \log |A|)$ and the runtime for answering each operation would be $O(1)$. 

\subsection{From Counting to Detecting Exact Locations}
We show a nice reduction from counting the number of solution indices to actually reporting the positions of the indices that are affected by a change. To make the explanation simpler, we introduce a toy problem and show the reduction in terms of the toy problem. Let $B: [1,m] \rightarrow \{0,1\}$ be a 0/1 array of size $m$ which is hidden from us. We would like to discover which places of array $B$ have value $1$ by asking \textit{count queries}. For each count query, we give a subset $S$ of indices in $[1,m]$ to an oracle, and in return, we learn how many indices $i$ of $S$ satisfy $B_i = 1$. The goal is to find out which indices of $B$ have value $1$ with a small number of queries to the oracle. Intuitively, the toy problem shows how to use count queries in order to find the location of the elements with value $1$.

For reasons that become clear later, we are interested in non-adaptive algorithms. That is, all the queries are made to the oracle before learning about the answers. However, let us start for simplicity by explaining an adaptive algorithm. We denote the total number of ones in array $B$ by $k$. The easiest case is when $k = 1$ in which case we can solve the problem with a binary search with $O(\log m)$ count queries. Even for larger $k$, we can still modify the binary search idea to detect all of the desired positions in time $O(k \log m)$. It is not hard to show that a lower bound of $\tilde \Omega(k)$ holds on the number of necessary count queries when $k$ is not very close to $m$ (say $k < m/2$).

The next step is to make the count queries \textit{non-adaptive}. In this case, we first give a list of queries to the oracle, and after receiving all of the queries, the oracle reports the answers of the queries to us. Obviously, binary search cannot be used in the non-adaptive setting. Let us again begin with the case of $k=1$ in the non-adaptive setting. We make $\lceil \log m \rceil+1$ queries in the following way: For any $0 \leq \gamma \leq \lceil \log m \rceil$, we make a query that includes every element $i$ such that the $\gamma$'th bit of the base-2 representation of $i$ is equal to $1$. Since $k=1$, the answer to each query will be either $0$ or $1$. Let $s$ be the index of $B$ such that $B_s = 1$. For a $\gamma$, if the answer to the $\gamma$'th query is equal to $1$, it means that the $\gamma$'th bit of the base-2 representation of $s$ is equal to $1$. Otherwise, we know that the $\gamma$'th bit of the base-2 representation of $s$ is equal to $0$. Thus, by looking at the answers of the $\lceil \log m \rceil+1$ queries, we can detect what is the value of $s$. Notice that these queries are completely independent and thus they can be made non-adaptively.

While the above idea solves the problem for $k=1$, unfortunately it does not generalize to larger $k$. Therefore, we bring a sampling technique which in addition to the above ideas solves the problem for us. To show the technique, we consider the case of $k=2$. In this case, if we are able to divide the elements of $B$ into two subsets $B^1$ and $B^2$ such that each of $B^1$ and $B^2$ contains one element of $B$ with value $1$, we could solve the problem separately for each subset and report the locations of both desired indices. If we were to make adaptive queries, finding such $B^1$ and $B^2$ would be easy: We would keep making a random decomposition of indices of $B$ into $B^1$ and $B^2$ and ask a query with subset $S = B^1$ until the answer to the query is equal to $1$. Since the choices of $B^1$ and $B^2$ are random, we know that w.h.p, we find such a decomposition after $O(\log m)$ queries. Once we find such a decomposition, we solve the problem subject to sets $B^1$ and $B^2$ separately using the ideas discussed previously. This keeps the number of queries $O(\log m)$ for $k=2$.

Unfortunately, this is not possible in the non-adaptive setting. Therefore, we need to consider all of the $O(\log m)$ random choices of $B^1$ and $B^2$ in our algorithm. That is, we make $O(\log m)$ random decompositions of the elements of $B$ into two subsets and for each of them we make a query to see how many solution elements lie in each set. Next, for each random decomposition, we assume that the two solution indices are divided between them and solve the problem separately subject to the two subsets. When all of the queries are given to the oracle, we can find out which decomposition meets our desired condition and based on that we can detect which places of the array have value $1$. This does solve the problem but with more queries in comparison to the adaptive setting; Instead of $O(\log m)$ queries that we would make in the adaptive setting, we now make $O(\log^2 m)$ queries.

This idea is generalizable to larger $k$ but the additional multiplicative factor is large. That is, if we decompose array $B$ into $k$ subsets $B^1, B^2, \ldots,B^k$, the odds that each one of the subsets has exactly one element with value $1$ decreases exponentially in terms of $k$. Therefore, instead of having only $O(\log m)$ random decompositions for $k=2$, we would need to have $k^{O(k)} \log m$ decompositions to make the algorithm work for larger $k$. However, we can remedy this issue in the following way: We first set our goal to only report one of the solution indices (chosen arbitrarily). In this case, instead of a decomposition, we just need to find a subset $B^*$ of $B$ such that $B^*$  only contains one element of $B$ with value $1$. If we put each element of $B$ into $B^*$ with probability $1/k$, we would only need $O(\log m)$ random choices of $B^*$ to make sure one choice is desirable with high probability. We can then use the binary coding ideas to detect the position of the corresponding element with $O(\log m)$ additional queries. Again, since the queries are non-adaptive, we would need to make all of the queries for all $O(\log m)$ random choices of $B^*$ which amounts to $O(\log^2 m)$ queries in total.
	
The above idea gives one of the solution indices to us via $O(\log^2 m)$ queries. To discover the rest of the solution indices, we repeat the procedure all over again, this time only for the $k-1$ remaining solution elements (thus the probability that an element is included in $B^*$ is $1/(k-1)$). The only difference is that whenever a query  contains the already discovered solution index, we manually decrease the reported value by $1$ to ignore that element. The crux of the argument is that since the sets are made completely at random, the queries are completely independent and can be asked non-adaptively. That is, we make $O(\log^2 m)$ queries for discovering the first solution index. The next $O(\log^2 m)$ queries are completely independent of the first $O(\log^2 m)$ queries.  If we follow this pattern $k$ times, we would be able to discover all  indices of $B$ with value $1$ with $O(k \log^2 m)$ queries.

We can now  make a connection between the above toy problem and our solution for static subset sum. Recall that the solution is represented by an array $A:[0,\tmax] \rightarrow \{0,1\}$. After computing the polynomial multiplication, we can use the computed sequence to find out how many elements of $A$ will be affected if we add an item with any weight. Now, for $m = \tmax+1$, consider an array $B$ where $B_i$ is equal to $1$, if and only if $A_{i-1} = 0$ and $A_{i-w-1} = 1$ (the shift in the indices is due to the fact that array $A$ starts at index 0 but array $B$ starts at index 1).

Counting the number of indices that are affected by a change is equivalent to counting the number of elements in $B$ that are equal to $1$. However, we can extend this notion to any subset of elements in $B$. Recall that our solution for the count queries involves making two 0/1 sequences $X$ and $Y$ and taking their polynomial multiplication. Now, if we only want to take into account a subset $S$ of elements of $A$, it suffices to set $Y_{|A|-i+1} = 0$ if $i \notin S$. Thus, when considering the corresponding toy problem that concerns array $B$, we can implement the count queries by the polynomial multiplication technique. Therefore, a reduction from counting to detecting the exact location of the indices carries over to static subset sum. That is, for a fixed $k$, one can design an algorithm that preprocess the input in time $\tilde O(k|A|)$ and is able to detect the changes in $A$ after an element with weight $w$ is added to the problem so long as the number of changes is bounded by $k$. Keep in mind that if the number of changes exceeds $k$, our algorithm is not able to detect any of the affected indices at all. Otherwise, the runtime for detecting each affected index is $\tilde O(k)$. In our technical sections, we refer to the solution of the static subset sum problem as the \textsf{$k$-flip-detector} data structure.

\subsection{From Static Subset Sum to Dynamic Subset Sum}
In this section, we consider the dynamic subset sum problem. For simplicity, we assume that the number of operations is equal to $\tmax$ ($\opers = \tmax$) but we show in Section~\ref{sec:dynamic} that the average processing time of our algorithm remains truly sublinear as long as $\opers \geq \tmax^{2/3+\Omega(1)}$. In our algorithm, we keep an array $A:[0,\tmax] \rightarrow \{0,1\}$ that represents our solution. $A_i = 1$ means that there is a subset of items whose total weight is equal to $i$ and $A_i = 0$ means otherwise. In the beginning, we set $A_0 \leftarrow 1$ and $A_i \leftarrow 0$ for $1 \leq i \leq \tmax$.
 
Our algorithm utilizes \textsf{$k$-flip-detector} to dynamically solve subset sum however, this requires us to address two issues. The first shortcoming of \textsf{$k$-flip-detector} is that it is designed in a static manner. That is, once we detect the changes to array $A$ after an element insertion, we cannot incorporate that into \textsf{$k$-flip-detector} to be able to answer future queries. The other shortcoming of \textsf{$k$-flip-detector} is that it is only capable of detecting the changes in $A$ when their count is bounded by $k$. Therefore, unless we set $k = O(\tmax)$, we are not able to detect all of the changes that an element addition may incur to $A$. Since the processing time of \textsf{$k$-flip-detector} is $\tilde O(k|A|) = \tilde O(k \tmax)$, setting $k = \tmax$ makes the preprocessing time quadratic which prevents us from obtaining a dynamic algorithm with truly sublinear processing time. 

We resolve the first issue by modifying \textsf{$k$-flip-detector} to obtain a new data structure namely \textsf{$k$-flip-detector*}. The stark difference between \textsf{$k$-flip-detector} and \textsf{$k$-flip-detector*} is that \textsf{$k$-flip-detector*} is able to change the original sequence after the preprocessing step. However, this comes with extra cost for answering queries. Let $\mathsf{cnt}$ be the number of changes made to the initial sequence after the preprocessing step of \textsf{$k$-flip-detector*}. We show that when we count the number of affected indices of $A$, we can incorporate the changes by spending additional time $\tilde O(\mathsf{cnt})$. Similarly, this comes with an extra cost of $\tilde O(\mathsf{cnt})$ for detecting each of the affected indices. We refer the reader to Section~\ref{sec:sub} for more details about \textsf{$k$-flip-detector*}.

We also resolve the second issue via the following trick: We set $k = \tmax^{1/3}$ and initialize the data structure on sequence $A$. After we process each query, we also update the data structure by incorporating the changes to $A$. We use the same data structure for up to $\tmax^{2/3}$ changes to $A$. However, each time the number of changes for an element addition exceeds $k$, instead of using \textsf{$k$-flip-detector*}, we spend time $O(\tmax)$ to discover all of the affected positions of $A$ explicitly. The crux of the argument is that since we only change the values of $A$ from 0 to $1$, the total number of times that happens is bounded by $\tmax / k = \tmax^{2/3}$. Therefore, throughout the life of our algorithm this amounts to a runtime of $\tilde O(\tmax^{5/3})$ and when taking average over all operations, this cost is only $\tilde O(\tmax^{2/3})$ per operation. We also keep track of changes to $A$ after each initialization and once the number of changes exceeds $\tmax^{2/3}$, we initialize the data structure from scratch by using the updated sequence. Since this happens at most $O(\tmax^{1/3})$ many times and each initialization takes time $\tilde O(k \tmax) = \tilde O(\tmax^{4/3})$, the total initialization cost is bounded by $\tilde O(\tmax^{5/3})$ which is $\tilde O(\tmax^{2/3})$ per operation. Also, since the number of changes to $A$ for each data structure never exceeds $\tmax^{2/3}$, the time required to count the number of changes and detect each change is also bounded by $\tilde O(\tmax^{2/3})$ for every operation.

While this gives amortized processing time $\tilde O(\tmax^{2/3})$ when $\tmax$ and $\opers$ are asymptotic, we show that by proper choice of $k$ and tuning other parameters, we can obtain a truly sublinear amortized processing time so long as the number of operations is at least $\tmax^{2/3+\Omega(1)}$. Finally, we would like to note that our algorithm is \textbf{not} aware of the number of operations in advance but manages to keep the amortized processing time bounded in terms of $\opers$. The full description of our algorithm and its analysis is given in Section~\ref{sec:dynamic}. 

\subsection{Dynamic Subset Sum with Offline Access}
The problem becomes more tractable if offline access is provided to the operations in advance. For this setting, we prove that we can process the operations in time $\tilde O(1+\tmax \min{\{\sqrt{\opers},\sqrt{\tmax}\}}/\opers)$ on average. This means that the processing time is truly sublinear so long as $\opers \geq \tmax^{\Omega(1)}$ which matches the lower bound implied from~\cite{abboud2019seth}.

\vspace{0.2cm}
{\noindent \textbf{Theorem}~\ref{thm:offline}, [restated informally]. \textit{	There is an algorithm for dynamic subset sum with offline access whose amortized processing time is bounded by  $\tilde O(1+\tmax \min{\{\sqrt{\opers},\sqrt{\tmax}\}}/\opers)$.\\}}

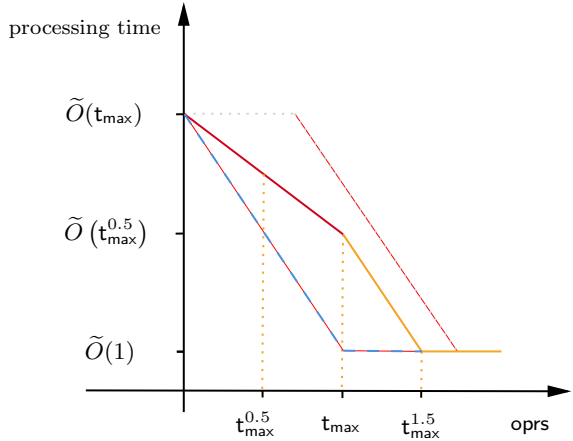
\begin{figure}[ht]
	\centering
	\input{figs/offline-time}
	\caption{The amortized processing time of our dynamic subset sum algorithm  with offline access is shown via solid lines. The red part is for the case $\opers < \tmax$ and the orange part is for the case $\opers \geq \tmax$. The blue dashed polyline shows the lower bound on the amortized processing time which directly follows from the lower bound of~\cite{abboud2019seth} for subset sum. The dotted polyline shows the performance of our dynamic algorithm that does not have offline access to the operations.}\label{fig:offline-time}
\end{figure}

Our algorithm takes advantage of the extra  information about the operations given in advance. Therefore, we first make an initialization in our algorithm and from then on we would be able to process each query in sublinear time. We assume for simplicity here that $\opers = \tmax$.

We divide the operations by $\milestones$ \textit{milestones} in our algorithm. Let the milestones be at operations $m_1, m_2, \ldots, m_{\lambda}$. The time steps between consecutive milestones are chosen in a way that the solution sequence changes in at most $\tmax / \milestones$ indices. We show in Section~\ref{sec:off} how we determine the milestones in time $\tilde O(\milestones \tmax)$.

 In our algorithm, before we process any of the operations, for each milestone $m_i$, we solve the problem from scratch. That is, we consider all of the items that are added prior to that milestone and find which numbers can be made as sum of weights of those items. This takes time  $\tilde {O}(\tmax)$ for each milestone using Bringmann's algorithm \cite{bringmann2017near} which in total amounts to $\tilde {O}(\lambda \tmax)$. After that, we make an array $A[0,\tmax] \rightarrow \{0,1\}$ such that $A_0 = 1$ and $A_i = 0 $ for all $1 \leq i \leq \tmax$ which represents the solution. We then process the operations one by one. If an operation adds an element we update $A$ and otherwise we report the solution by looking at $A$.

When we reach a milestone, say $m_i$, updating $A$ would be trivial; we just look at the solution which is precomputed for that milestone and based on that we update $A$. For operations between two consecutive milestones, say $m_i$ and $m_{i+1}$, we know that $A$ only changes for at most $\tmax/\lambda$ indices. Due to our preprocessing, those indices are also available to us. Therefore, all it takes in order to process each insertion operation is to iterate over those indices and find out whether the solution changes for any of them. Thus, the processing time of each insertion operation is $O(\tmax /\lambda)$. Moreover, the query operations take time $O(1)$ to be processed. Thus, the overall runtime of our algorithm would be $\tilde O(\lambda \tmax + \tmax \opers /\lambda) = \tilde O(\lambda \tmax +  \tmax^2 /\lambda)$ and by setting $\lambda = \sqrt{\tmax}$ we obtain an amortized runtime of $\tilde O(\sqrt{\tmax})$ per operation.

For a larger number of operations ($\opers > \tmax$), we make one more observation that further improves the runtime: for a weight $w$, if at least $\tmax / w$ items with weight $w$ are already added to the problem, we can ignore the rest of the operations that insert an item with weight $w$. This implies that we can narrow down the set of insertion operations to at most $O(\tmax \log \tmax)$ many operations. This enables us to improve the runtime down to $\tilde O(\tmax^{1.5}/\opers)$ in this case. More details about the algorithm is given in Section~\ref{sec:off}.

\subsection{Special Cases: Bounded 3-sum and Bounded $k$-sum}
We also consider special cases of subset sum, namely bounded 3-sum and more generally bounded $k$-sum. We prove in Section~\ref{sec:3sum} that any dynamic algorithm for subset sum yields a dynamic algorithm with the same processing time for bounded $3$-sum. This carries over to the bounded $k$-sum problem except that the average processing time will be multiplied by a function of $k$. 

\vspace{0.2cm}
{\noindent \textbf{Theorem}~\ref{sum-reduction}, [restated informally]. \textit{For any constant $k \geq 3$, (fully) dynamic subset sum is at least as hard as (fully) dynamic bounded $k$-sum.\\}}

We proceed  to design \textbf{fully} dynamic algorithms for bounded 3-sum and bounded $k$-sum in Section~\ref{sec:3sum} that have truly sublinear processing times in terms of $\rmax$.

\vspace{0.2cm}
{\noindent \textbf{Theorem}~\ref{thm:3sum}, [restated informally]. \textit{
		There is an algorithm for fully dynamic $3$-sum whose amortized processing time is bounded by  $\tilde O(\min{\{\opers, \rmax^{0.5}\}})$. \\}}
	
We start by presenting a fully dynamic algorithm for bounded $3$-sum with sublinear amortized processing time. %The main results of this section are summarized in Theorem~\ref{thm:3sum} and Corollary~\ref{thm:3sum-addition}.
In our algorithm, we keep an auxiliary dataset which we refresh after every $\sqrt{\rmax}$ operations. We call such an event a \textit{milestone}. We also maintain a list $R$ which is a pile of operations which have come after the most recent computation of the auxiliary dataset. Therefore the size of $R$ never exceeds $\sqrt{\rmax}$. We also maintain a number $\textsf{cnt}$ throughout the algorithm such that at every point in time, $\textsf{cnt}$ counts the number of tuples $(a,b,c)$ such that $a \in \Aa, b \in \Bb, c \in \Cc$ and $a+b = c$. The ultimate goal is to update the value of $\textsf{cnt}$ after each operation so that we can answer query operations based on whether $\mathsf{cnt} > 0$ or not. In the beginning, and also after every $\sqrt{\rmax}$ operations, we compute the auxiliary dataset. More precisely, we compute 
\[\big\{\Aa, \Bb, -\Cc, \Aa+\Bb, \Aa-\Cc, \Bb - \Cc\big\}. \]
where $\Aa+\Bb$ (or $\Aa-\Bb$) is a set containing $a+b$ (or $a-b$) for every pair of elements $(a,b) \in \Aa \times\Bb$. Similarly, $-\Cc$ contains the negation of each number in $\Cc$.

Note that we treat each of the above data as a \textbf{non binary} vector over range $[-2\rmax, 2\rmax]$ where every element specifies in how many ways a specific number can be made. For instance $\Aa + \Bb$ is a vector over range $[-2\rmax,2\rmax]$ and index $i$ counts the number of pairs $(a,b)$ such that $a \in \Aa$ and $b \in \Bb$ and $a+b = i$. It follows that each of the vectors can be computed in time $O(\rmax \log \rmax)$ using polynomial multiplication. When the number of elements in the sets is small, an alternative way to compute $\Aa + \Bb$ is iterating over all pairs which can be done in time $O(\opers^2)$.
Each time we compute the auxiliary dataset, we also start over with a new pile $R$ with no operations in it. Moreover, after the computation of the auxiliary dataset, we also compute $\textsf{cnt}$ using polynomial multiplication in time $O(\rmax \log \rmax)$. The only exception is the first time we make the auxiliary dataset in which case we know $\mathsf{cnt}$ is equal to 0 and we do not spend any time on computing it.

The auxiliary dataset remains intact until we recompute it from scratch but as new operations are added into $R$, we need to update $\textsf{cnt}$. To update $\textsf{cnt}$ after an insertion operation or a deletion operation, we need to compute the number of triples that affect $\textsf{cnt}$ and contain the newly added (or deleted) number. To explain the idea, let us assume that a number $w$ is added into $\Aa$. There are four types of triples $(w,b,c)$ that can potentially affect $\textsf{cnt}$: (Keep in mind that $b$ and $c$ may refer to some numbers that previously existed in the sets but were removed at some point.)

\begin{itemize}
	\item Both $b$ and $c$ have arrived before the last milestone and thus they are incorporated in it. The number of such pairs is $\big(\Bb-\Cc\big)_{-w}$.
	\item $b$ is added or deleted after the last milestone, but $c$ has come before the last milestone. In this case, we iterate over all new modifications to $\Bb$, which their count is bounded by $|R|$. For each of them we look up $(-\Cc)_{-w-b}$ to verify if there is a triple containing $w$ and $b$. %We need to subtract the count of new modifications that delete item $b'$ while $b'$ existed in $\Bb$ during the last milestone.
	\item In this case, $b$ refers to an operation prior to the last milestone and $c$ refers to an operation after the last milestone. This case is similar to the previous case and we just iterate over new modifications to $\Cc$.
	\item In this case, both $b$ and $c$ refer to operations after the last milestone. The number of such pairs can also be computed in time $O(|R|)$ since by fixing $b$, we just search to verify if $c = w+b$ exists in the new set of operations.
\end{itemize}

Depending on whether or not an operation adds a number or removes a number, we may incorporate the number of tuples into $\textsf{cnt}$ positively or negatively. However, since all the above counts are available in time $O(|R|)$, we can update $\mathsf{cnt}$ in time $O(|R|)$. This implies that the amortized update time of our algorithm is bounded by $\tilde O(\min{\{\opers, \rmax^{0.5}\}})$. We also generalize this algorithm to work for fully dynamic bounded $k$-sum.

\vspace{0.2cm}
{\noindent \textbf{Theorem}~\ref{thm:ksum}, [restated informally]. \textit{
		There is an algorithm for fully dynamic $k$-sum whose amortized processing time is bounded by  $\tilde O(\min{\{\opers^{k-2}, \rmax^{(k-2)/(k-1)}\}})$. \\}}
\begin{comment}
\begin{figure}[ht]
	\centering
	\input{figs/3sum}
	\caption{The amortized processing time of our (fully) dynamic $3$-sum algorithm is shown via solid lines.	}\label{fig:3sum}
	%The blue dashed polyline shows the lower bound on the amortized processing time which directly follows from the lower bound Abboud et al.~\cite{abboud2019seth} for subset sum.
\end{figure}
\end{comment}

%% file: figs/online-time.tex
\tikzset{every picture/.style={line width=0.75pt}} %set default line width to 0.75pt        

\begin{tikzpicture}[x=0.75pt,y=0.75pt,yscale=-1,xscale=1]
	%uncomment if require: \path (0,662); %set diagram left start at 0, and has height of 662
	
	%Straight Lines [id:da09485126686390766] 
	\draw [color={rgb, 255:red, 245; green, 166; blue, 35 }  ,draw opacity=1 ] [dash pattern={on 0.84pt off 2.51pt}]  (287,249.97) -- (287,269.6) ;
	%Straight Lines [id:da5013604796457034] 
\draw [color={rgb, 255:red, 245; green, 166; blue, 35 }  ,draw opacity=1 ] [dash pattern={on 0.84pt off 2.51pt}]  (206.74,130.54) -- (206.74,269.92) ;
%Straight Lines [id:da5454391623954071] 
\draw [color={rgb, 255:red, 245; green, 166; blue, 35 }  ,draw opacity=1 ] [dash pattern={on 0.84pt off 2.51pt}]  (230.33,170.84) -- (230.33,270.04) ;

%Straight Lines [id:da5454391623954071] 
\draw [color={rgb, 255:red, 245; green, 166; blue, 35 }  ,draw opacity=1 ] [dash pattern={on 0.84pt off 2.51pt}]  (150,170) -- (230.33,170) ;

	%Straight Lines [id:da31977965221169913] 
	\draw    (100.6,269.64) -- (343.5,270.09) ;
	\draw [shift={(345.5,270.09)}, rotate = 180.11] [fill={rgb, 255:red, 0; green, 0; blue, 0 }  ][line width=0.08]  [draw opacity=0] (12,-3) -- (0,0) -- (12,3) -- cycle    ;
	%Straight Lines [id:da7091936608091394] 
	\draw    (150,280.59) -- (150,88.97) -- (150,75.59) ;
	\draw [shift={(150,73.59)}, rotate = 450] [fill={rgb, 255:red, 0; green, 0; blue, 0 }  ][line width=0.08]  [draw opacity=0] (12,-3) -- (0,0) -- (12,3) -- cycle    ;
	%Straight Lines [id:da7504399820828549] 
	\draw    (230.02,269.78) -- (230.02,273.49) ;
	%Straight Lines [id:da4411592394743922] 
	\draw    (286.74,269.6) -- (286.74,273.32) ;
	%Straight Lines [id:da5273189801809726] 
	\draw    (150.4,130.24) -- (146.13,130.24) ;
	%Straight Lines [id:da784465713499481] 
	\draw    (150.15,170) -- (145.88,170) ;
	%Straight Lines [id:da05028785782489553] 
	\draw [color={rgb, 255:red, 245; green, 166; blue, 35 }  ,draw opacity=1 ][fill={rgb, 255:red, 255; green, 0; blue, 0 }  ,fill opacity=1 ]   (206,130) -- (288,249.97) ;
	%Straight Lines [id:da07379180928094375] 
	\draw [color={rgb, 255:red, 245; green, 166; blue, 35 }  ,draw opacity=1 ]   (288,249.97) -- (310.25,249.97) ;
	%Straight Lines [id:da7244652501844011] 
	\draw    (146.5,250.09) -- (150,250.09) ;
	%Straight Lines [id:da6868351574871974] 
	\draw [color={rgb, 255:red, 208; green, 2; blue, 27 }  ,draw opacity=1 ]   (150.4,130) -- (206,130) ;

	%Straight Lines [id:da4991084863046744] 
	\draw    (206.74,269.49) -- (206.74,273.21) ;
	%Straight Lines [id:da033748997346433995] 
	\draw [color={rgb, 255:red, 74; green, 144; blue, 226 }  ,draw opacity=1 ][fill={rgb, 255:red, 255; green, 0; blue, 0 }  ,fill opacity=1 ] [dash pattern={on 4.5pt off 4.5pt}]  (150.4,130.24) -- (230.5,249.59) ;
	%Straight Lines [id:da5904727823170357] 
	\draw [color={rgb, 255:red, 74; green, 144; blue, 226 }  ,draw opacity=1 ][fill={rgb, 255:red, 255; green, 0; blue, 0 }  ,fill opacity=1 ] [dash pattern={on 4.5pt off 4.5pt}]  (230.5,249.59) -- (288,249.97) ;
	
	% Text Node
	\draw (313.5,283.27) node [anchor=north west][inner sep=0.75pt]    {$\opers$};
	% Text Node
	\draw (225.9,282.18) node [anchor=north west][inner sep=0.75pt]  [font=\footnotesize]  {$\tmax$};
	% Text Node
	\draw (278.29,277.77) node [anchor=north west][inner sep=0.75pt]  [font=\footnotesize]  {$\tmax^{5/3}$};
	% Text Node
	\draw (84.14,160.34) node [anchor=north west][inner sep=0.75pt]  [font=\footnotesize]  {$\widetilde{O}\left(\tmax^{2/3}\right)$};
	% Text Node
	\draw (88.79,120.34) node [anchor=north west][inner sep=0.75pt]  [font=\footnotesize]  {$\widetilde{O}( \tmax)$};
	% Text Node
	\draw (60.3,81.14) node [anchor=north west][inner sep=0.75pt]  [font=\scriptsize]  {$\text{processing time}$};
	% Text Node
	\draw (101.64,240.84) node [anchor=north west][inner sep=0.75pt]  [font=\footnotesize]  {$\widetilde{O}( 1)$};
	% Text Node
	\draw (190.76,278.89) node [anchor=north west][inner sep=0.75pt]  [font=\footnotesize]  {$\tmax^{2/3}$};
\end{tikzpicture}

%% file: figs/offline-time.tex
\tikzset{every picture/.style={line width=0.75pt}} %set default line width to 0.75pt        

\begin{tikzpicture}[x=0.75pt,y=0.75pt,yscale=-1,xscale=1]
	%uncomment if require: \path (0,662); %set diagram left start at 0, and has height of 662
	
		\draw [color={rgb, 255:red, 210; green, 210; blue, 210 }  ,draw opacity=1 ][dash pattern={on 0.84pt off 2.51pt}] [fill={rgb, 255:red, 255; green, 0; blue, 0 }  ,fill opacity=1 ]   (206,130) -- (288,249.97) ;
	%Straight Lines [id:da07379180928094375] 
%	\draw [color={rgb, 255:red, 190; green, 190; blue, 190 }  ,draw opacity=1 ] [dash pattern={on 0.84pt off 2.51pt}]   (288,249.97) -- (310.25,249.97) ;
	%Straight Lines [id:da7244652501844011] 
	\draw    (146.5,250.09) -- (150,250.09) ;
	%Straight Lines [id:da6868351574871974] 
	\draw [color={rgb, 255:red, 210; green, 210; blue, 210 }  ,draw opacity=1 ] [dash pattern={on 0.84pt off 2.51pt}]   (150.4,130) -- (206,130) ;
	
	%Straight Lines [id:da09485126686390766] 
	\draw [color={rgb, 255:red, 245; green, 166; blue, 35 }  ,draw opacity=1 ] [dash pattern={on 0.84pt off 2.51pt}]  (270,249.97) -- (269.74,269.6) ;
	%Straight Lines [id:da5454391623954071] 
	\draw [color={rgb, 255:red, 245; green, 166; blue, 35 }  ,draw opacity=1 ] [dash pattern={on 0.84pt off 2.51pt}]  (230.33,190.84) -- (229.72,270.04) ;
	%Straight Lines [id:da31977965221169913] 
	\draw    (100.6,269.64) -- (343.5,270.09) ;
	\draw [shift={(345.5,270.09)}, rotate = 180.11] [fill={rgb, 255:red, 0; green, 0; blue, 0 }  ][line width=0.08]  [draw opacity=0] (12,-3) -- (0,0) -- (12,3) -- cycle    ;
	%Straight Lines [id:da7091936608091394] 
	\draw    (150,280.59) -- (150,88.97) -- (150,75.59) ;
	\draw [shift={(150,73.59)}, rotate = 450] [fill={rgb, 255:red, 0; green, 0; blue, 0 }  ][line width=0.08]  [draw opacity=0] (12,-3) -- (0,0) -- (12,3) -- cycle    ;
	%Straight Lines [id:da7504399820828549] 
	\draw    (230.02,269.78) -- (230.02,273.49) ;
	%Straight Lines [id:da4411592394743922] 
	\draw    (269.74,269.6) -- (269.74,273.32) ;
	%Straight Lines [id:da5273189801809726] 
	\draw    (150.4,130.24) -- (146.13,130.28) ;
	%Straight Lines [id:da784465713499481] 
	\draw    (150.15,190.47) -- (145.88,190.51) ;
	%Straight Lines [id:da05028785782489553] 
	\draw [color={rgb, 255:red, 245; green, 166; blue, 35 }  ,draw opacity=1 ][fill={rgb, 255:red, 255; green, 0; blue, 0 }  ,fill opacity=1 ]   (230.33,190.84) -- (270,249.97) ;
	%Straight Lines [id:da07379180928094375] 
	\draw [color={rgb, 255:red, 245; green, 166; blue, 35 }  ,draw opacity=1 ]   (270,249.97) -- (310.25,249.97) ;
	%Straight Lines [id:da7244652501844011] 
	\draw    (146.5,250.09) -- (150,250.09) ;
	%Straight Lines [id:da6868351574871974] 
	\draw [color={rgb, 255:red, 208; green, 2; blue, 27 }  ,draw opacity=1 ]   (150.4,130.24) -- (230.33,190.84) ;
	%Straight Lines [id:da5013604796457034] 
	\draw [color={rgb, 255:red, 245; green, 166; blue, 35 }  ,draw opacity=1 ] [dash pattern={on 0.84pt off 2.51pt}]  (190.37,160.54) -- (189.57,269.92) ;
	%Straight Lines [id:da4991084863046744] 
	\draw    (189.74,269.49) -- (189.74,273.21) ;
	%Straight Lines [id:da033748997346433995] 
	\draw [color={rgb, 255:red, 74; green, 144; blue, 226 }  ,draw opacity=1 ][fill={rgb, 255:red, 255; green, 0; blue, 0 }  ,fill opacity=1 ] [dash pattern={on 4.5pt off 4.5pt}]  (150.4,130.24) -- (230.5,249.59) ;
	%Straight Lines [id:da5904727823170357] 
	\draw [color={rgb, 255:red, 74; green, 144; blue, 226 }  ,draw opacity=1 ][fill={rgb, 255:red, 255; green, 0; blue, 0 }  ,fill opacity=1 ] [dash pattern={on 4.5pt off 4.5pt}]  (230.5,249.59) -- (270,249.97) ;
	
	% Text Node
	\draw (313.5,283.27) node [anchor=north west][inner sep=0.75pt]   [font=\scriptsize] {$\opers$};
	% Text Node
	\draw (217.9,281.18) node [anchor=north west][inner sep=0.75pt]  [font=\footnotesize]  {$\tmax$};
	% Text Node
	\draw (258.29,279.77) node [anchor=north west][inner sep=0.75pt]  [font=\footnotesize]  {$\tmax^{1.5}$};
	% Text Node
	\draw (88.14,180.34) node [anchor=north west][inner sep=0.75pt]  [font=\footnotesize]  {$\widetilde{O}\left(\tmax^{0.5}\right)$};
	% Text Node
	\draw (88.79,120.34) node [anchor=north west][inner sep=0.75pt]  [font=\footnotesize]  {$\widetilde{O}( \tmax)$};
	% Text Node
	\draw (60.3,81.14) node [anchor=north west][inner sep=0.75pt]  [font=\scriptsize]  {$\text{processing time}$};
	% Text Node
	\draw (98.64,240.84) node [anchor=north west][inner sep=0.75pt]  [font=\footnotesize]  {$\widetilde{O}( 1)$};
	% Text Node
	\draw (174.76,278.89) node [anchor=north west][inner sep=0.75pt]  [font=\footnotesize]  {$\tmax^{0.5}$};

\end{tikzpicture}

%% file: src/main.tex
\section{Dynamic Subset Sum With Sublinear Processing Time}\label{sec:dynamic}
In this section, we present our main result which is a dynamic algorithm for subset sum whose amortized processing time per operation is truly sublinear when the number of  operations is at least $\tmax^{2/3+\Omega(1)}$. We obtain this result through multiple combinatorial ideas married with the classic application of Fast Fourier Transform to design an efficient data structure for subset sum. We begin by providing a formal statement of our result in Theorem~\ref{theorem:main}.

\begin{theorem}\label{theorem:main}
	There exists an algorithm for dynamic subset sum whose amortized processing time is bounded by $O(1 + \tmax^{5/3} \log^2 \tmax \log \opers/\opers)$. The algorithm answers each query operation correctly with probability at least $1-1/\tmax^3$.
\end{theorem}

Theorem~\ref{theorem:main} does not make use of existing subquadratic time algorithms for subset sum and thus can be seen as an independent improved algorithm for subset sum which features flexibility with respect to incremental changes. In other words, the result of Theorem~\ref{theorem:main} leads to an algorithm for subset sum which runs in time $O(n+t^{5/3})$. (Notice that although the runtime is subquadratic in $n+t$, this algorithm is still slower than the more advanced algorithms given for subset sum~\cite{koiliaris2019faster, bringmann2017near}.)

Although the amortized processing time of our algorithm is based on parameter $\opers$,  our algorithm is not aware of this value in advance. Instead, we maintain a parameter $\hat{\opers}$ which 2-approximates the number of operations. We start by setting $\hat{\opers} \leftarrow 1$ and maintain its value intact until the number of processed operations reaches $\hat{\opers}$. Once we reach that threshold, we multiply $\hat{\opers}$ by 2 and run the algorithm all over again. That is, we start anew and reprocess all of the already processed operations and then continue processing new operations from then on. We discuss at the end of the section that this comes with an additional overhead of $O(\log \opers)$ to the average processing time of the algorithm. In the below discussion, we assume w.l.o.g that $\opers$ remains intact throughout the life of the algorithm.

The processing time of our algorithm is independent of the number of items at the time the operations arrive. In fact, our algorithm keeps track of items only because it requires to reprocess the operations each time $\hat{\opers}$ changes. Other than that, our algorithm is oblivious to the items that have come prior to any operation and only takes into account the weight of the newly added item to the problem.

The algorithm is naive when $\hat{\opers} < \tmax^{2/3}$; We keep a 0/1 array of $A$ of size $|\tmax+1|$ which stores the solution for various values of $t$ and after each query operation, we report the answer in time $O(1)$. In this case, the processing time of our algorithm is $O(\tmax)$ per operation. In what follows, we consider $\tmax^{2/3} \leq \hat{\opers} \leq \tmax$ and show how we obtain a sublinear processing time in this case. Later in the section, we discuss how to make the algorithm work when $\hat{\opers} > \tmax$.

In our algorithm, we maintain a 0/1 array $A: [0,\tmax] \rightarrow \{0,1\}$ such that $A_i$ always tells us whether there exists a subset of items whose sum of weights is exactly equal to $i$. Thus, the query operations always take time $O(1)$. The challenging part of the algorithm is to efficiently update array $A$ after every insertion operation. For that purpose, we design a data structure which is parametrized by a variable $k$ namely, \textit{\textsf{$k$-flip-detector*}}. \textsf{$k$-flip-detector*} initializes on a 0/1 array $Z$ and is able to answer queries of the following form: (i) For a given parameter $\alpha$, it reports the number of indices of $Z$ that meet Condition~\eqref{eq:pro}. (ii) For a given parameter $\alpha$ for which no more than $k$ indices of $Z$ meet Condition \eqref{eq:pro}, it will report all such indices. 
\begin{equation}\label{eq:pro}
\alpha \leq i \leq |Z| \hspace{1cm}\text{ and } \hspace{1cm}  Z_i = 0 \hspace{1cm}\text{ and } \hspace{1cm} Z_{i-\alpha} = 1.
\end{equation}
If more than $k$ indices of $Z$ meet Condition \eqref{eq:pro}, \textsf{$k$-flip-detector*} may not give us any meaningful information about their location. Moreover, \textsf{$k$-flip-detector*} supports changes to the initial sequence $Z$. The time complexities of the operations supported by \textsf{$k$-flip-detector*} is as follows:\\[0.5cm]

\fbox{\begin{minipage}{40em}
\vspace{0.3cm}
\hspace{0.5cm}\textsf{\textsf{$k$-flip-detector*}}:\\[-0.4cm]
\begin{itemize}
	\item Initialization of \textsf{$k$-flip-detector*} on a 0/1 array $Z$ takes time $O(k|Z| \log^2 |Z|)$.
	\item Each time we flip an index of $Z$, it takes  $O(1)$ time to update \textsf{$k$-flip-detector*}.
	\item For a given $\alpha$, it takes time $O(\mathsf{cnt})$ to determine how many elements of $Z$ meet Condition \eqref{eq:pro} where $\mathsf{cnt}$ is the number of changes made to $Z$ after its initialization.
	\item If for a given $\alpha$, the number of indices of $Z$ that meet Condition \eqref{eq:pro} is bounded by $q \leq k$, the runtime for detecting all of such indices is $O((q+\mathsf{cnt})q \log |Z|)$. That is, per solution index the algorithm takes time $O((q+\mathsf{cnt}) \log |Z|)$.\\[0.3cm]
\end{itemize}

\end{minipage}}

\vspace{0.5cm}
 
We defer the details of \textsf{$k$-flip-detector*} to Section~\ref{sec:sub} and focus on the overall algorithm here. (A formal statement of the time complexities of \textsf{$k$-flip-detector*} is given in Lemma~\ref{lemma:2}.) For a parameter $1 \leq \Delta$ whose value will be determined at the end of the section, we set $k = \Delta \tmax / \hat{\opers}$ and update array $A$ in the following way. Initially, we set $A_0 \leftarrow 1$ and $A_i \leftarrow 0$ for all $i \in [1, \tmax]$. Also in the beginning, we make an instance of \textsf{$k$-flip-detector*} and initialize it on array $A$. Whenever an operation adds an item with value $w$, we query \textsf{$k$-flip-detector*} to find how many elements of $A$ are affected by adding the new item (these are basically the elements that meet Condition~\eqref{eq:pro} for $\alpha = w$). If their count is more than $k = \Delta \tmax / \hat{\opers}$, we iterate over the entire array $A$ to find the positions of all such changes. If their count is bounded by $k = \Delta \tmax / \hat{\opers}$, we query \textsf{$k$-flip-detector*} to list all those positions. In any case, we update array $A$ accordingly.

We also need to update \textsf{$k$-flip-detector*} as $A$ changes in our algorithm. Every time an element of $A$ changes, we perform the same modification to \textsf{$k$-flip-detector*}. However, when the number of changes exceeds $\hat{\opers}/\Delta$, instead of modifying \textsf{$k$-flip-detector*}, we initialize it from scratch. This takes time $O(k\tmax \log^2 \tmax)$ but it resets the number of edits to zero. Keep in mind that since elements of $A$ only change from $0$ to $1$, the total number of changes is bounded by $\tmax$ and thus we initialize \textsf{$k$-flip-detector*} at most $\tmax / (\hat{\opers}/\Delta) = \Delta \tmax/\hat{\opers}$ times. Thus, the total time required for all initializations combined is 
\begin{equation*}
\begin{split}
(\Delta  \tmax/\hat{\opers})  O( k\tmax \log^2 \tmax) &=  O((\Delta  \tmax /\hat{\opers}) k\tmax \log^2 \tmax)\\
&=  O((\Delta  \tmax/\hat{\opers}) (\Delta \tmax / \hat{\opers})\tmax \log^2 \tmax)\\
& =  O(\Delta^2\tmax^3 \log^2 \tmax/\hat{\opers}^2 ).
\end{split}
\end{equation*}
Thus, per operation, we spend average time $O(\Delta^2 \tmax^3 \log^2 \tmax/\hat{\opers}^3)$ for initializations.  Finally, since we always limit the number of edits to  \textsf{$k$-flip-detector*} by $\hat{\opers}/\Delta$ and $k=\Delta \tmax / \hat{\opers}$, the runtime of queries of type (ii) to \textsf{$k$-flip-detector*} per change to $A$ is bounded by: 
\begin{equation*}
\begin{split}
O((q+\mathsf{cnt}) \log \tmax) & \leq O((k+\mathsf{cnt}) \log \tmax)\\
&\leq O((\Delta \tmax / \hat{\opers}+\hat{\opers}/\Delta) \log \tmax)
\end{split}
\end{equation*}
where $q$ denotes the solution size of a query of type (ii) and $\textsf{cnt}$ denotes the number of changes after initialization. Therefore, the total cost for all changes is bounded by $O((\Delta \tmax / \hat{\opers}+\hat{\opers}/\Delta) \tmax \log \tmax)$. Thus, per operation, we pay a cost of $O((\Delta \tmax / \hat{\opers}^2+1/\Delta) \tmax \log \tmax)$. There is one more cost which is negligible for $\opers \leq \tmax$. Every time a new item is added to the problem, we query \textsf{$k$-flip-detector*} to find out how many elements of $A$ are affected by the new item. This takes time $O(\textsf{cnt}) \leq O(\hat{\opers}/\Delta)$ per operation. 

In total, the amortized processing time per operation of our algorithm is bounded by $O(\Delta^2 \tmax^3 \log^2 \tmax/\hat{\opers}^3) + O((\Delta \tmax / \hat{\opers}^2+1/\Delta) \tmax \log \tmax) + O(\hat{\opers}/\Delta)$. When $\tmax^{2/3} \leq \hat{\opers} \leq \tmax$, by setting $\Delta = \hat{\opers}/\tmax^{2/3}$ we can obtain a dynamic algorithm whose amortized processing time is $O(\tmax \log^2 \tmax/\Delta) = O(\tmax^{5/3} \log^2 \tmax/\hat{\opers})$. 

For the cases that $\opers > \tmax$, we make the following modification to the algorithm. We also keep track of how many items of each weight are added to the problem. Moreover, after an item of weight $w$ is added $\lfloor \tmax / w \rfloor$ times, we ignore all the coming operations that add an item with weight $w$. This does not change the outcome since we never use more than $\lfloor \tmax / w \rfloor$ items of weight $w$ in any solution. Thus, the total number of times that we make type (i) queries to the data structure would be bounded by $\tmax + \lfloor \tmax /2 \rfloor + \lfloor \tmax /4 \rfloor + \ldots + 1 =  O(\tmax \log \tmax)$. Thus, the total time we spend on queries of type (i) would be $O((\tmax \log \tmax)(\hat{\opers}/\Delta))$.  Therefore, per operation, the amortize processing time of our algorithm would be $O(\Delta^2 \tmax^3 \log^2 \tmax/\hat{\opers}^3) + O((\Delta \tmax / \hat{\opers}^2+1/\Delta) \tmax \log \tmax) + O((\tmax \log \tmax)/\Delta)$. Again, by setting by setting $\Delta = \hat{\opers}/\tmax^{2/3}$ we can obtain a dynamic algorithm whose amortized processing time is $O(\tmax \log^2 \tmax/\Delta) = O(\tmax^{5/3} \log^2 \tmax/\hat{\opers})$.

\begin{proof}[of Theorem~\ref{theorem:main}]
	We explained the description of the algorithm and time complexities earlier. Here we just argue why there is an additional overhead of $O(\log \opers)$ in the processing time of the algorithm. Since the average processing time of the algorithm is $O(\tmax^{5/3} \log^2 \tmax/\hat{\opers})$, each time we multiply $\hat{\opers}$ by $2$, the total time we spend on all operations will be bounded by $O(\tmax^{5/3} \log^2 \tmax)$. Since the number of times we multiply $\hat{\opers}$ by $2$ is bounded by $\log \opers$, the overall runtime over all operations would be $O(\tmax^{5/3} \log^2 \tmax \log \opers)$. Thus, per operation we spend time $O(1+\tmax^{5/3} \log^2 \tmax \log \opers/\opers)$. There is an additional additive $+1$ term in the runtime since we spend an additional time $O(1)$ on each query for receiving it.
	
	In terms of correctness, we point out that each query of type (ii) to \textsf{$k$-flip-detector*} may fail with probability at most $1/\tmax^4$ and since the number of those queries never exceeds $\tmax$, our algorithm always reports a correct solution with probability at least $1-\tmax^3$.
\end{proof}

\input{src/data-structure}

%% file: src/data-structure.tex
\subsection{\textsf{$k$-flip-detector*}}\label{sec:sub}
The goal of this section is to present our algorithm for \textsf{$k$-flip-detector*}. To this end, we begin by a simpler data structure, namely \textsf{$k$-asymmetric-flip-detector} and then show a reduction from  \textsf{$k$-flip-detector*} to \textsf{$k$-asymmetric-flip-detector}. In \textsf{$k$-asymmetric-flip-detector}, we are given two 0/1 sequences $X$ and $Y$ and the goal to preprocess the sequences in a way that our algorithm would be able to answer queries of the following forms efficiently: (i) Given an $1-|X| \leq \alpha \leq |Y|-1$, how many indices $i$ of $Y$ meet Condition~\eqref{condition:con}?

\begin{equation}\label{condition:con}
1 \leq i \leq |Y|, 1 \leq i-\alpha \leq |X| \hspace{1cm}\text{ and } \hspace{1cm}  Y_i = 0 \hspace{1cm}\text{ and } \hspace{1cm} X_{i-\alpha} = 1.
\end{equation}

(ii) Also in case for some $\alpha$, the answer to the type (i) query is bounded by $k$, \textsf{$k$-asymmetric-flip-detector} should also be able to detect such positions. We begin by explaining how type (i) queries can be answered and then we move on to type (ii) queries. This can be thought of as a special case of \textsf{$k$-asymmetric-flip-detector} for the case of $k=0$.

It is well known that for any 0/1 sequences $A$ and $B$, one can spend time $O((|A|+|B|) \log (|A|+|B|))$ and construct a (not necessarily 0/1) sequence $C$ of size $|A|+|B|$ such that for any $2 \leq i \leq |A|+|B|$ we have $C_i = \sum_{1 \leq j \leq |A|, 1 \leq i-j \leq |B|} A_j B_{i-j}$. This is a classic application of Fast Fourier Transform to the polynomial multiplication of 0/1 sequences. 

\textsf{$0$-asymmetric-flip-detector} can be solved using the above technique. In the preprocessing step, we construct a 0/1 sequence $\hat{Y}$ of size $|Y|$ such that $\hat{Y}_i = 1-Y_{|Y|-i+1}$. In other words, we change $Y$ in two ways to obtain $\hat{Y}$: We first reverse it, and then we flip each of its elements. We then use the Fast Fourier Transform technique to construct a sequence of size $|X|+|\hat{Y}|$, namely $R$ such that for any $2 \leq i \leq |X|+|\hat{Y}|$ we have $R_i = \sum_{1 \leq j \leq |X|, 1 \leq i-j \leq |\hat{Y}|} X_j \hat{Y}_{i-j}$. This basically completes the preprocessing needed for \textsf{$0$-asymmetric-flip-detector}.

When a query of type (i) with parameter $\alpha$ is given to \textsf{$0$-asymmetric-flip-detector}, the answer is already computed in sequence $R$. More precisely, for $i = |Y|+1 - \alpha$ we have:
\begin{equation*}
\begin{split}
R_i & = \sum_{1 \leq j \leq |X|, 1 \leq i-j \leq |\hat{Y}|} X_j \hat{Y}_{i-j} \\
& = \sum_{1 \leq j \leq |X|, 1 \leq i-j \leq |\hat{Y}|} X_j (1-Y_{|Y|-(i-j)+1})\\
& = \sum_{1 \leq j \leq |X|, 1 \leq |Y|+1 - \alpha-j \leq |\hat{Y}|} X_j (1-Y_{|Y|-(|Y|+1 - \alpha-j)+1})\\
& = \sum_{1 \leq j \leq |X|, 1 \leq |Y|+1 - \alpha-j \leq |\hat{Y}|} X_j (1-Y_{j+\alpha})\\
& = \sum_{1 \leq j \leq |X|, -|Y| \leq - \alpha-j \leq -1} X_j (1-Y_{j+\alpha})\\
& = \sum_{1 \leq j \leq |X|, 1 \leq j+\alpha \leq |Y|} X_j (1-Y_{j+\alpha})
\end{split}
\end{equation*}
and thus $R_i$ correctly gives us the answer for query type (i) with parameter $\alpha$.

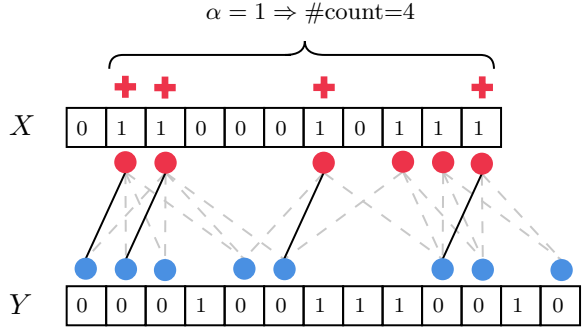
\begin{figure}[ht]
	\centering
	\input{figs/flip-detector}
	\caption{An example of $k\textsf{-asymmetric-flip-detector}$ for $k=0$. Two binary vectors $X$ and $Y$ are given  and we wish to count the number of indices $i$ so that $Y_i = 0$ and $X_{i-\alpha}=1$ for a given parameter $\alpha$.}\label{fig:flip-detector}
\end{figure}

In the above algorithm, the preprocessing time is $O((|X|+|Y|) \log (|X|+|Y|))$ and each query can be answered in time $O(1)$. We now extend the above idea to present a solution for \textsf{$1$-asymmetric-flip-detector}. In the case of \textsf{$1$-asymmetric-flip-detector}, when for some $\alpha$ the solution to type (i) query is exactly equal to $1$, we should be able to also find the corresponding index.

As part of the preprocessing step of \textsf{$1$-asymmetric-flip-detector}, we make sequences $\hat{Y}$ and $R$. In addition to that, we construct $l+1 = \lceil \log |Y| \rceil+1$ auxiliary sequences $\bar{Y}_0, \bar{Y}_1, \bar{Y}_2, \ldots, \bar{Y}_l$ each of size $|Y|$. For each $0 \leq \gamma \leq l$, $\bar{Y}_\gamma$ is formulated in the following way:
\begin{itemize}
	\item $\forall 1 \leq i \leq |Y|$ such that the $\gamma$'th bit of the base 2 representation of $|Y|-i+1$ is equal to 0, we have $(\bar{Y}_\gamma)_i = 0$. This implies that $Y_{|Y|-i+1}$ will not be incorporated in the construction of $\bar{Y}_\gamma$.
	\item Otherwise we have $(\bar{Y}_\gamma)_i = \hat{Y}_i$.
\end{itemize}
In other words, $\bar{Y}_\gamma$ is exactly the same as $\hat{Y}$ except that it only takes into account the indices of $Y$ whose $\gamma$'th bit in the base-2 representation is equal to $1$. Similarly, we make $l+1$ additional sequences $\bar{R}_0, \bar{R}_1, \bar{R}_2, \ldots, \bar{R}_l$ in the same way we make $R$. The only difference is that instead of computing the polynomial multiplication of $X$ and $\hat{Y}$, $\bar{R}_\gamma$ takes the multiplication of $X$ and $\bar{Y}_\gamma$.

Similar to \textsf{$0$-asymmetric-flip-detector}, \textsf{$1$-asymmetric-flip-detector} can also answer queries of type (i) in $O(1)$ time. Now, consider a query of type (ii) with parameter $\alpha$ such that its answer for the type (i) query is exactly equal to $1$. This means that there is a unique index $s$ in a way that $X_{s-\alpha}  = 1$ and $Y_s = 0$. Moreover, we already know that $R_{|Y|+1-\alpha} = 1$ holds.

To discover $s$, we take into account sequences $\bar{R}_0, \ldots, \bar{R}_l$. For any  $0 \leq \gamma \leq l$, it follows from the construction of $\bar{Y}_\gamma$ that $(\bar{R}_\gamma)_{|Y|+1-\alpha} = 1$ if and only if the $\gamma$'th bit of the base-2 representation of $s$ is equal to $1$. This means that by putting the values of $(\bar{R}_\gamma)_{|Y|+1-\alpha}$ next to each other, we get the base-2 representation of $s$. Thus, we can answer type (ii) queries in $O(\log |Y|)$ time. 

Unfortunately, the above idea does not generalize to $k > 1$ for \textsf{$k$-asymmetric-flip-detector}. For $k > 1$, we make $10k \log |Y|$ instances of \textsf{$1$-asymmetric-flip-detector} and put them in baskets of size $10 \log |Y|$. We also label the baskets with numbers $1$ to $k$. We use the same initialization algorithm for each instance, except that for an instance in basket $i$, we take into account each index of $Y$ with probability $1/i$. In other words, for a pair of initialization sequences $X$ and $Y$ we initialize an instance of basket $i$ in the following way:
\begin{itemize}
	\item We make a 0/1 sequence $Y'$ which is initially equal to $Y$.
	\item We iterate over its indices and turn each element into $1$ with probability $(i-1)/i$. For all such indices $Y'_j$ that are intentionally set to $1$, we say the corresponding \textsf{$1$-asymmetric-flip-detector} does not take into account the $j$'th element of $Y$. We also store which indices are taken into account for that instance.
	\item We initialize the instance based on sequences $X$ and $Y'$.
\end{itemize}
There is a redundancy in the basket with label 1 that is all of its \textsf{$1$-asymmetric-flip-detector}s are all made the same way but for the sake of simplicity, we do not optimize our algorithm for that particular basket. The intuition behind the algorithm is the following: Let for a query of type (ii) with value $\alpha$, $t$ denote the number of indices of $Y$ which meet Condition~\eqref{condition:con}. If $t \leq k$, with high probability, there is at least one instance of \textsf{$1$-asymmetric-flip-detector} in basket $t$ that takes into account exactly one of those indices. Based on the ideas above, we can detect that particular index in time $O(\log |Y|)$ in the following way: We iterate over all \textsf{$1$-asymmetric-flip-detector}s in basket $t$ and find the one which reports $1$ to us when making a query type (i) with value $\alpha$. We then make a query type (ii) with value $\alpha$ to that \textsf{$1$-asymmetric-flip-detector}.

We then move on to the basket with label $t-1$. With the same argument, with high probability, there is at least one instance of \textsf{$1$-asymmetric-flip-detector} in that basket, that takes into account exactly one of the $t-1$ remaining indices of $Y$ that meet Condition~\eqref{condition:con}. However, when giving $\alpha$ to that instance as a query type (i), the answer may not be equal to $1$ since the solution index that we have detected already may also be incorporated in that \textsf{$1$-asymmetric-flip-detector}. (In that case, the value reported for query type (i) will be equal to $2$.) Thus, we need to manually update the answers: If the already detected index is incorporated in a \textsf{$1$-asymmetric-flip-detector}, we decrease the reported value by $1$. This also applies to the additional queries we make in order to locate the solution indices. Therefore, the second index that meets Condition~\eqref{condition:con} can also be found in time $O(\log |Y|)$.

When trying to locate the third index, we move to basket $t-2$ but this time we need to refine the reported values by considering both of the detected solution indices. This still gives us runtime $O(\log |Y|)$ but as we detect more solution indices the runtime increases proportionally to the number of discovered indices. Therefore, there is an additional $O(t)$ overhead in the runtime of the algorithm and thus the overall runtime for detecting all such indices will be $O(t^2 \log |Y|)$ which is $O(t \log |Y|)$ per each index. This completes the implementation of \textsf{$k$-asymmetric-flip-detector}.

\begin{lemma}\label{lemma:0}
	\textsf{$k$-asymmetric-flip-detector} preprocesses the inputs in time $O(k(|X|+|Y|) \log^2 (|X|+|Y|))$ and can answer queries of type (i) in time $O(1)$ and queries of type (ii) in time $O(t^2 \log |Y|)$  when the number of solution indices is equal to $t \leq k$. The answers to queries of type (ii) are correct with probability at least $1-1/|Y|^4$.
\end{lemma}
\begin{proof}
	We discussed the algorithms and time complexities above. Here we just bound the failure probability. When there are $t$ indices that meet Condition~\eqref{condition:con},  an instance of \textsf{$1$-asymmetric-flip-detector} in basket $t$ will take exactly one of such indices into account with probability 
	$$t((1/t) (\frac{t-1}{t})^{t-1}) = (\frac{t-1}{t})^{t-1} \geq 1/e > 0.367.$$
	Therefore, among the $10 \log |Y|$ instances of basket $t$, with probability at least $1-1/|Y|^5$, there is at least one instance that takes exactly one of the solution indices into account. Since we repeat this proceduter $t$ times to detect all the solution indices, the failure probability will be multiplied by an additional factor $t$. Therefore, our algorithm answers each query of type (ii) correctly with probability at least $1-t/|Y|^5 \geq 1-1/|Y|^4$. Also, type (i) queries are always reported correctly since no randomness is used in their algorithm.
\end{proof}

The next step is to modify \textsf{$k$-asymmetric-flip-detector} to make it flexible with respect to changes to initial sequences $X$ and $Y$. We call such a data structure \textsf{$k$-asymmetric-flip-detecto*}. \textsf{$k$-asymmetric-flip-detecto*} is the same as \textsf{$k$-asymmetric-flip-detector} except that in addition to the two query types, it will also be able to handle changes to initial arrays $X$ and $Y$ in the following form: Each change comes either in terms of flipping one element of $X$ or one element of $Y$. Our algorithm handles these changes in time $O(1)$ but these changes will incur an overhead to the runtime of answering query types (i) and (ii).

The overall idea is simple: We use an instance of \textsf{$k$-asymmetric-flip-detector} but we also keep track of the changes that are made to the input sequence after initialization. Once a query of type (i) with parameter $\alpha$ is given, we make the query to \textsf{$k$-asymmetric-flip-detector} to find the answer had the changes not been made to the sequence. We then iterate over the changes and see how the changes would affect the outcome. For every change, it takes time $O(1)$ to see if this affects the outcome: If an element of $X$, say $X_i$, is modified, we look at $Y_{i+\alpha}$ and verify if this change affects the outcome. If so, we update the solution accordingly. Thus, this incurs a multiplicative overhead to the runtime of answering queries of types (i) which is proportional to the number of changes.

Queries of type (ii) are also answered in the same way. However, when the number of solution indices is equal to $t \leq k$, instead of having an overhead $O(t)$ for fixing the reported values, the overhead would be equal to $O(t+\mathsf{cnt})$ where $\mathsf{cnt}$ is the number of changes made to the data structure after its initialization.

\begin{lemma}(as a Corollary of Lemma~\ref{lemma:0})\label{lemma:1}
	\textsf{$k$-asymmetric-flip-detector*} preprocesses the inputs in time $O(k(|X|+|Y|) \log^2 (|X|+|Y|))$. Each modification to \textsf{$k$-asymmetric-flip-detector*} takes time $O(1)$. When the number of changes after the initialization is equal to $\mathsf{cnt}$, \textsf{$k$-asymmetric-flip-detector*} can answer queries of type (i) in time $O(\mathsf{cnt})$ and queries of type (ii) in time $O((t+\mathsf{cnt})t \log |Y|)$ when the number of solution indices is equal to $t \leq k$. The answers to queries of type (ii) are correct with probability at least $1-1/|Y|^4$.
\end{lemma}

We are now ready to discuss the details of \textsf{$k$-flip-detector*}. Recall that in this data structure, an initial 0/1 array $Z$ is given and there are two types of queries. In type (i), for a value $\alpha$, our algorithm needs to find out how many elements of $Z$ meet Condition~\eqref{eq:pro} (For clarity, we restate Condition~\eqref{eq:pro} below). For type (ii) queries, we are given an $\alpha$ such that the solution of type (i) query with value $\alpha$ is bounded by $k$ and we need to detect all the positions that meet Condition~\eqref{eq:pro}.

\begin{equation}
\alpha \leq i \leq |Z| \hspace{1cm}\text{ and } \hspace{1cm}  Z_i = 0 \hspace{1cm}\text{ and } \hspace{1cm} Z_{i-\alpha} = 1. \tag{\ref{eq:pro}}
\end{equation}

The reduction from \textsf{$k$-flip-detector*} to \textsf{$k$-asymmetric-flip-detector*} is straightforward. For a given 0/1 sequence $Z$, we construct a \textsf{$k$-asymmetric-flip-detector*} with both parameters $X = Z$ and $Y=Z$. Moreover, once an element of $Z$ changes, we report the two corresponding changes to the \textsf{$k$-asymmetric-flip-detector*} instance. Whenever a query of type (i) or a query of type (ii) arrives, we directly forward that to the \textsf{$k$-asymmetric-flip-detector*} and report the answer to the output.

\begin{lemma}(as a Corollary of Lemma~\ref{lemma:1})\label{lemma:2}
	\textsf{$k$-flip-detector*} preprocesses the inputs in time $O(k|Z| \log^2 |Z|)$. Each modification to \textsf{$k$-asymmetric-flip-detector*} takes time $O(1)$. When the number of changes after the initialization is equal to $\mathsf{cnt}$, \textsf{$k$-flip-detector*} can answer queries of type (i) in time $O(\mathsf{cnt})$ and queries of type (ii) in time $O((t+\mathsf{cnt})t \log |Z|)$ when the number of solution indices is equal to $t \leq k$. The answers to queries of type (ii) are correct with probability at least $1-1/|Z|^4$.
\end{lemma}

%% file: figs/flip-detector.tex
\tikzset{every picture/.style={line width=0.75pt}} %set default line width to 0.75pt        

\begin{tikzpicture}[x=0.75pt,y=0.75pt,yscale=-1,xscale=1]
	%uncomment if require: \path (0,662); %set diagram left start at 0, and has height of 662
	
	%Straight Lines [id:da979881815586036] 
	\draw [color={rgb, 255:red, 200; green, 200; blue, 200 }  ,draw opacity=1 ] [dash pattern={on 4.5pt off 4.5pt}]  (120.13,482.72) -- (120.38,526.59) ;
	%Straight Lines [id:da1206002980723615] 
	\draw [color={rgb, 255:red, 200; green, 200; blue, 200 }  ,draw opacity=1 ] [dash pattern={on 4.5pt off 4.5pt}]  (120.13,482.72) -- (140.13,527.03) ;
	%Straight Lines [id:da2137306919876103] 
	\draw [color={rgb, 255:red, 200; green, 200; blue, 200 }  ,draw opacity=1 ] [dash pattern={on 4.5pt off 4.5pt}]  (140.38,482.97) -- (100.13,526.34) ;
	%Straight Lines [id:da7449706929808795] 
	\draw [color={rgb, 255:red, 200; green, 200; blue, 200 }  ,draw opacity=1 ] [dash pattern={on 4.5pt off 4.5pt}]  (140.38,482.97) -- (140.13,527.03) ;
	%Straight Lines [id:da15669726642776105] 
	\draw [color={rgb, 255:red, 200; green, 200; blue, 200 }  ,draw opacity=1 ] [dash pattern={on 4.5pt off 4.5pt}]  (220.13,482.97) -- (180.13,526.34) ;
	%Straight Lines [id:da9839978579074251] 
	\draw [color={rgb, 255:red, 200; green, 200; blue, 200 }  ,draw opacity=1 ] [dash pattern={on 4.5pt off 4.5pt}]  (260.13,482.47) -- (280.13,526.59) ;
	%Straight Lines [id:da8819003682743602] 
	\draw [color={rgb, 255:red, 200; green, 200; blue, 200 }  ,draw opacity=1 ] [dash pattern={on 4.5pt off 4.5pt}]  (260.13,482.47) -- (300.38,526.84) ;
	%Straight Lines [id:da6638097991644927] 
	\draw [color={rgb, 255:red, 200; green, 200; blue, 200 }  ,draw opacity=1 ] [dash pattern={on 4.5pt off 4.5pt}]  (280.38,482.72) -- (280.13,526.59) ;
	%Straight Lines [id:da04502449119821539] 
	\draw [color={rgb, 255:red, 200; green, 200; blue, 200 }  ,draw opacity=1 ] [dash pattern={on 4.5pt off 4.5pt}]  (280.38,482.72) -- (300.38,526.84) ;
	%Straight Lines [id:da8928591833683643] 
	\draw [color={rgb, 255:red, 200; green, 200; blue, 200 }  ,draw opacity=1 ] [dash pattern={on 4.5pt off 4.5pt}]  (280.38,482.72) -- (340.13,527.03) ;
	%Straight Lines [id:da7502644020093576] 
	\draw [color={rgb, 255:red, 200; green, 200; blue, 200 }  ,draw opacity=1 ] [dash pattern={on 4.5pt off 4.5pt}]  (299.88,483.47) -- (300.38,526.84) ;
	%Straight Lines [id:da16710817155025404] 
	\draw [color={rgb, 255:red, 200; green, 200; blue, 200 }  ,draw opacity=1 ] [dash pattern={on 4.5pt off 4.5pt}]  (299.88,483.47) -- (340.13,527.03) ;
	%Straight Lines [id:da8809306839105915] 
	\draw [color={rgb, 255:red, 200; green, 200; blue, 200 }  ,draw opacity=1 ] [dash pattern={on 4.5pt off 4.5pt}]  (140.38,482.97) -- (180.13,526.34) ;
	%Straight Lines [id:da11474979550896536] 
	\draw [color={rgb, 255:red, 200; green, 200; blue, 200 }  ,draw opacity=1 ] [dash pattern={on 4.5pt off 4.5pt}]  (120.13,482.72) -- (180.13,526.34) ;
	%Straight Lines [id:da6136238033852885] 
	\draw [color={rgb, 255:red, 200; green, 200; blue, 200 }  ,draw opacity=1 ] [dash pattern={on 4.5pt off 4.5pt}]  (140.38,482.97) -- (200.13,526.59) ;
	%Straight Lines [id:da5242916296037439] 
	\draw [color={rgb, 255:red, 200; green, 200; blue, 200 }  ,draw opacity=1 ] [dash pattern={on 4.5pt off 4.5pt}]  (220.13,482.97) -- (280.13,526.59) ;
	%Straight Lines [id:da20452586309132514] 
	\draw [color={rgb, 255:red, 200; green, 200; blue, 200 }  ,draw opacity=1 ] [dash pattern={on 4.5pt off 4.5pt}]  (260.13,482.47) -- (200.38,526.59) ;
	%Straight Lines [id:da06737971328215653] 
	\draw    (299.88,483.47) -- (279.88,527.09) ;
	%Straight Lines [id:da3848541844700524] 
	\draw    (220.13,482.97) -- (200.13,526.59) ;
	%Straight Lines [id:da8871006925582563] 
	\draw    (140.38,482.97) -- (120.38,526.59) ;
	%Straight Lines [id:da6111241081306473] 
	\draw    (120.13,482.72) -- (100.13,526.34) ;
	%Shape: Rectangle [id:dp8423720745280814] 
	\draw   (90.5,450.25) -- (110.42,450.25) -- (110.42,469.97) -- (90.5,469.97) -- cycle ;
	%Shape: Rectangle [id:dp19157743742162636] 
	\draw   (110.42,450.25) -- (130.33,450.25) -- (130.33,469.97) -- (110.42,469.97) -- cycle ;
	%Shape: Rectangle [id:dp7789809023283389] 
	\draw   (130.33,450.25) -- (150.25,450.25) -- (150.25,469.97) -- (130.33,469.97) -- cycle ;
	%Shape: Rectangle [id:dp24183827808438085] 
	\draw   (150.25,450.25) -- (170.17,450.25) -- (170.17,469.97) -- (150.25,469.97) -- cycle ;
	%Shape: Rectangle [id:dp07706622729182322] 
	\draw   (170.17,450.25) -- (190.08,450.25) -- (190.08,469.97) -- (170.17,469.97) -- cycle ;
	%Shape: Rectangle [id:dp9379690798491502] 
	\draw   (190.08,450.25) -- (210,450.25) -- (210,469.97) -- (190.08,469.97) -- cycle ;
	%Shape: Rectangle [id:dp8025915190810113] 
	\draw   (210,450.25) -- (229.92,450.25) -- (229.92,469.97) -- (210,469.97) -- cycle ;
	%Shape: Rectangle [id:dp21685950918640473] 
	\draw   (229.92,450.25) -- (249.83,450.25) -- (249.83,469.97) -- (229.92,469.97) -- cycle ;
	%Shape: Rectangle [id:dp03593750440444432] 
	\draw   (249.83,450.25) -- (269.75,450.25) -- (269.75,469.97) -- (249.83,469.97) -- cycle ;
	%Shape: Rectangle [id:dp10881691240056846] 
	\draw   (269.75,450.25) -- (289.67,450.25) -- (289.67,469.97) -- (269.75,469.97) -- cycle ;
	%Shape: Rectangle [id:dp03407524035809684] 
	\draw   (289.67,450.25) -- (309.58,450.25) -- (309.58,469.97) -- (289.67,469.97) -- cycle ;
	%Shape: Rectangle [id:dp32322136370398447] 
	\draw   (90.5,540.25) -- (110.42,540.25) -- (110.42,559.97) -- (90.5,559.97) -- cycle ;
	%Shape: Rectangle [id:dp18178407274718777] 
	\draw   (110.42,540.25) -- (130.33,540.25) -- (130.33,559.97) -- (110.42,559.97) -- cycle ;
	%Shape: Rectangle [id:dp4275665363583854] 
	\draw   (130.33,540.25) -- (150.25,540.25) -- (150.25,559.97) -- (130.33,559.97) -- cycle ;
	%Shape: Rectangle [id:dp254207850869558] 
	\draw   (150.25,540.25) -- (170.17,540.25) -- (170.17,559.97) -- (150.25,559.97) -- cycle ;
	%Shape: Rectangle [id:dp2876407381918453] 
	\draw   (170.17,540.25) -- (190.08,540.25) -- (190.08,559.97) -- (170.17,559.97) -- cycle ;
	%Shape: Rectangle [id:dp14771141734940474] 
	\draw   (190.08,540.25) -- (210,540.25) -- (210,559.97) -- (190.08,559.97) -- cycle ;
	%Shape: Rectangle [id:dp7443935661033052] 
	\draw   (210,540.25) -- (229.92,540.25) -- (229.92,559.97) -- (210,559.97) -- cycle ;
	%Shape: Rectangle [id:dp418436043135753] 
	\draw   (229.92,540.25) -- (249.83,540.25) -- (249.83,559.97) -- (229.92,559.97) -- cycle ;
	%Shape: Rectangle [id:dp23510204704100546] 
	\draw   (249.83,540.25) -- (269.75,540.25) -- (269.75,559.97) -- (249.83,559.97) -- cycle ;
	%Shape: Rectangle [id:dp2280867608841779] 
	\draw   (269.75,540.25) -- (289.67,540.25) -- (289.67,559.97) -- (269.75,559.97) -- cycle ;
	%Shape: Rectangle [id:dp4148293046710234] 
	\draw   (289.67,540.25) -- (309.58,540.25) -- (309.58,559.97) -- (289.67,559.97) -- cycle ;
	%Shape: Rectangle [id:dp722815773204295] 
	\draw   (309.58,540.25) -- (329.5,540.25) -- (329.5,559.97) -- (309.58,559.97) -- cycle ;
	%Shape: Rectangle [id:dp269508326595695] 
	\draw   (329.5,540.25) -- (349.42,540.25) -- (349.42,559.97) -- (329.5,559.97) -- cycle ;
	%Shape: Ellipse [id:dp9005789557106441] 
	\draw  [color={rgb, 255:red, 237; green, 51; blue, 74 }  ,draw opacity=1 ][fill={rgb, 255:red, 237; green, 51; blue, 74 }  ,fill opacity=1 ] (115,477.66) .. controls (115,474.86) and (117.29,472.59) .. (120.13,472.59) .. controls (122.96,472.59) and (125.25,474.86) .. (125.25,477.66) .. controls (125.25,480.45) and (122.96,482.72) .. (120.13,482.72) .. controls (117.29,482.72) and (115,480.45) .. (115,477.66) -- cycle ;
	%Shape: Ellipse [id:dp7985484453286036] 
	\draw  [color={rgb, 255:red, 237; green, 51; blue, 74 }  ,draw opacity=1 ][fill={rgb, 255:red, 237; green, 51; blue, 74 }  ,fill opacity=1 ] (135.25,477.91) .. controls (135.25,475.11) and (137.54,472.84) .. (140.38,472.84) .. controls (143.21,472.84) and (145.5,475.11) .. (145.5,477.91) .. controls (145.5,480.7) and (143.21,482.97) .. (140.38,482.97) .. controls (137.54,482.97) and (135.25,480.7) .. (135.25,477.91) -- cycle ;
	%Shape: Ellipse [id:dp11891398258113539] 
	\draw  [color={rgb, 255:red, 237; green, 51; blue, 74 }  ,draw opacity=1 ][fill={rgb, 255:red, 237; green, 51; blue, 74 }  ,fill opacity=1 ] (255,477.41) .. controls (255,474.61) and (257.29,472.34) .. (260.13,472.34) .. controls (262.96,472.34) and (265.25,474.61) .. (265.25,477.41) .. controls (265.25,480.2) and (262.96,482.47) .. (260.13,482.47) .. controls (257.29,482.47) and (255,480.2) .. (255,477.41) -- cycle ;
	%Shape: Ellipse [id:dp29979307479108486] 
	\draw  [color={rgb, 255:red, 237; green, 51; blue, 74 }  ,draw opacity=1 ][fill={rgb, 255:red, 237; green, 51; blue, 74 }  ,fill opacity=1 ] (275.25,477.66) .. controls (275.25,474.86) and (277.54,472.59) .. (280.38,472.59) .. controls (283.21,472.59) and (285.5,474.86) .. (285.5,477.66) .. controls (285.5,480.45) and (283.21,482.72) .. (280.38,482.72) .. controls (277.54,482.72) and (275.25,480.45) .. (275.25,477.66) -- cycle ;
	%Shape: Ellipse [id:dp3754805506000285] 
	\draw  [color={rgb, 255:red, 237; green, 51; blue, 74 }  ,draw opacity=1 ][fill={rgb, 255:red, 237; green, 51; blue, 74 }  ,fill opacity=1 ] (215,477.91) .. controls (215,475.11) and (217.29,472.84) .. (220.13,472.84) .. controls (222.96,472.84) and (225.25,475.11) .. (225.25,477.91) .. controls (225.25,480.7) and (222.96,482.97) .. (220.13,482.97) .. controls (217.29,482.97) and (215,480.7) .. (215,477.91) -- cycle ;
	%Shape: Ellipse [id:dp9814819972041089] 
	\draw  [color={rgb, 255:red, 237; green, 51; blue, 74 }  ,draw opacity=1 ][fill={rgb, 255:red, 237; green, 51; blue, 74 }  ,fill opacity=1 ] (294.75,478.41) .. controls (294.75,475.61) and (297.04,473.34) .. (299.88,473.34) .. controls (302.71,473.34) and (305,475.61) .. (305,478.41) .. controls (305,481.2) and (302.71,483.47) .. (299.88,483.47) .. controls (297.04,483.47) and (294.75,481.2) .. (294.75,478.41) -- cycle ;
	%Shape: Ellipse [id:dp7880888244146043] 
	\draw  [color={rgb, 255:red, 74; green, 144; blue, 226 }  ,draw opacity=1 ][fill={rgb, 255:red, 74; green, 144; blue, 226 }  ,fill opacity=1 ] (95,531.41) .. controls (95,528.61) and (97.29,526.34) .. (100.13,526.34) .. controls (102.96,526.34) and (105.25,528.61) .. (105.25,531.41) .. controls (105.25,534.2) and (102.96,536.47) .. (100.13,536.47) .. controls (97.29,536.47) and (95,534.2) .. (95,531.41) -- cycle ;
	%Shape: Ellipse [id:dp6523158939398379] 
	\draw  [color={rgb, 255:red, 74; green, 144; blue, 226 }  ,draw opacity=1 ][fill={rgb, 255:red, 74; green, 144; blue, 226 }  ,fill opacity=1 ] (115.25,531.66) .. controls (115.25,528.86) and (117.54,526.59) .. (120.38,526.59) .. controls (123.21,526.59) and (125.5,528.86) .. (125.5,531.66) .. controls (125.5,534.45) and (123.21,536.72) .. (120.38,536.72) .. controls (117.54,536.72) and (115.25,534.45) .. (115.25,531.66) -- cycle ;
	%Shape: Ellipse [id:dp0086424073446012] 
	\draw  [color={rgb, 255:red, 74; green, 144; blue, 226 }  ,draw opacity=1 ][fill={rgb, 255:red, 74; green, 144; blue, 226 }  ,fill opacity=1 ] (175,531.41) .. controls (175,528.61) and (177.29,526.34) .. (180.13,526.34) .. controls (182.96,526.34) and (185.25,528.61) .. (185.25,531.41) .. controls (185.25,534.2) and (182.96,536.47) .. (180.13,536.47) .. controls (177.29,536.47) and (175,534.2) .. (175,531.41) -- cycle ;
	%Shape: Ellipse [id:dp04599629175784958] 
	\draw  [color={rgb, 255:red, 74; green, 144; blue, 226 }  ,draw opacity=1 ][fill={rgb, 255:red, 74; green, 144; blue, 226 }  ,fill opacity=1 ] (195.25,531.66) .. controls (195.25,528.86) and (197.54,526.59) .. (200.38,526.59) .. controls (203.21,526.59) and (205.5,528.86) .. (205.5,531.66) .. controls (205.5,534.45) and (203.21,536.72) .. (200.38,536.72) .. controls (197.54,536.72) and (195.25,534.45) .. (195.25,531.66) -- cycle ;
	%Shape: Ellipse [id:dp39965228256728014] 
	\draw  [color={rgb, 255:red, 74; green, 144; blue, 226 }  ,draw opacity=1 ][fill={rgb, 255:red, 74; green, 144; blue, 226 }  ,fill opacity=1 ] (275,531.66) .. controls (275,528.86) and (277.29,526.59) .. (280.13,526.59) .. controls (282.96,526.59) and (285.25,528.86) .. (285.25,531.66) .. controls (285.25,534.45) and (282.96,536.72) .. (280.13,536.72) .. controls (277.29,536.72) and (275,534.45) .. (275,531.66) -- cycle ;
	%Shape: Ellipse [id:dp8314691051242047] 
	\draw  [color={rgb, 255:red, 74; green, 144; blue, 226 }  ,draw opacity=1 ][fill={rgb, 255:red, 74; green, 144; blue, 226 }  ,fill opacity=1 ] (295.25,531.91) .. controls (295.25,529.11) and (297.54,526.84) .. (300.38,526.84) .. controls (303.21,526.84) and (305.5,529.11) .. (305.5,531.91) .. controls (305.5,534.7) and (303.21,536.97) .. (300.38,536.97) .. controls (297.54,536.97) and (295.25,534.7) .. (295.25,531.91) -- cycle ;
	%Shape: Ellipse [id:dp7946964600927] 
	\draw  [color={rgb, 255:red, 74; green, 144; blue, 226 }  ,draw opacity=1 ][fill={rgb, 255:red, 74; green, 144; blue, 226 }  ,fill opacity=1 ] (135,532.09) .. controls (135,529.3) and (137.29,527.03) .. (140.13,527.03) .. controls (142.96,527.03) and (145.25,529.3) .. (145.25,532.09) .. controls (145.25,534.89) and (142.96,537.16) .. (140.13,537.16) .. controls (137.29,537.16) and (135,534.89) .. (135,532.09) -- cycle ;
	%Shape: Ellipse [id:dp150469737959116] 
	\draw  [color={rgb, 255:red, 74; green, 144; blue, 226 }  ,draw opacity=1 ][fill={rgb, 255:red, 74; green, 144; blue, 226 }  ,fill opacity=1 ] (335,532.09) .. controls (335,529.3) and (337.29,527.03) .. (340.13,527.03) .. controls (342.96,527.03) and (345.25,529.3) .. (345.25,532.09) .. controls (345.25,534.89) and (342.96,537.16) .. (340.13,537.16) .. controls (337.29,537.16) and (335,534.89) .. (335,532.09) -- cycle ;
	%Shape: Brace [id:dp2658933717860126] 
	\draw   (306,430.51) .. controls (306,425.84) and (303.67,423.51) .. (299,423.51) -- (220.62,423.51) .. controls (213.95,423.51) and (210.62,421.18) .. (210.62,416.51) .. controls (210.62,421.18) and (207.29,423.51) .. (200.62,423.51)(203.62,423.51) -- (119,423.51) .. controls (114.33,423.51) and (112,425.84) .. (112,430.51) ;
	%Shape: Cross [id:dp9877042269746814] 
	\draw  [color={rgb, 255:red, 237; green, 51; blue, 74 }  ,draw opacity=1 ][fill={rgb, 255:red, 237; green, 51; blue, 74 }  ,fill opacity=1 ] (118.91,434.54) -- (121.29,434.54) -- (121.29,438.25) -- (125,438.25) -- (125,440.94) -- (121.29,440.94) -- (121.29,444.64) -- (118.91,444.64) -- (118.91,440.94) -- (115.2,440.94) -- (115.2,438.25) -- (118.91,438.25) -- cycle ;
	%Shape: Cross [id:dp4648057429949335] 
	\draw  [color={rgb, 255:red, 237; green, 51; blue, 74 }  ,draw opacity=1 ][fill={rgb, 255:red, 237; green, 51; blue, 74 }  ,fill opacity=1 ] (138.91,434.94) -- (141.29,434.94) -- (141.29,438.65) -- (145,438.65) -- (145,441.34) -- (141.29,441.34) -- (141.29,445.04) -- (138.91,445.04) -- (138.91,441.34) -- (135.2,441.34) -- (135.2,438.65) -- (138.91,438.65) -- cycle ;
	%Shape: Cross [id:dp9725164724042561] 
	\draw  [color={rgb, 255:red, 237; green, 51; blue, 74 }  ,draw opacity=1 ][fill={rgb, 255:red, 237; green, 51; blue, 74 }  ,fill opacity=1 ] (219.11,434.84) -- (221.49,434.84) -- (221.49,438.55) -- (225.2,438.55) -- (225.2,441.24) -- (221.49,441.24) -- (221.49,444.94) -- (219.11,444.94) -- (219.11,441.24) -- (215.4,441.24) -- (215.4,438.55) -- (219.11,438.55) -- cycle ;
	%Shape: Cross [id:dp0874252187561313] 
	\draw  [color={rgb, 255:red, 237; green, 51; blue, 74 }  ,draw opacity=1 ][fill={rgb, 255:red, 237; green, 51; blue, 74 }  ,fill opacity=1 ] (298.91,434.84) -- (301.29,434.84) -- (301.29,438.55) -- (305,438.55) -- (305,441.24) -- (301.29,441.24) -- (301.29,444.94) -- (298.91,444.94) -- (298.91,441.24) -- (295.2,441.24) -- (295.2,438.55) -- (298.91,438.55) -- cycle ;
	
	% Text Node
	\draw (94,455.34) node [anchor=north west][inner sep=0.75pt]  [font=\scriptsize]  {$0$};
	% Text Node
	\draw (114,455.34) node [anchor=north west][inner sep=0.75pt]  [font=\scriptsize]  {$1$};
	% Text Node
	\draw (133.85,455.34) node [anchor=north west][inner sep=0.75pt]  [font=\scriptsize]  {$1$};
	% Text Node
	\draw (154,455.34) node [anchor=north west][inner sep=0.75pt]  [font=\scriptsize]  {$0$};
	% Text Node
	\draw (174,455.349) node [anchor=north west][inner sep=0.75pt]  [font=\scriptsize]  {$0$};
	% Text Node
	\draw (194.25,455.34) node [anchor=north west][inner sep=0.75pt]  [font=\scriptsize]  {$0$};
	% Text Node
	\draw (214,455.34) node [anchor=north west][inner sep=0.75pt]  [font=\scriptsize]  {$1$};
	% Text Node
	\draw (254,455.34) node [anchor=north west][inner sep=0.75pt]  [font=\scriptsize]  {$1$};
	% Text Node
	\draw (273.75,455.34) node [anchor=north west][inner sep=0.75pt]  [font=\scriptsize]  {$1$};
	% Text Node
	\draw (294,455.34) node [anchor=north west][inner sep=0.75pt]  [font=\scriptsize]  {$1$};
	% Text Node
	\draw (234.25,455.59) node [anchor=north west][inner sep=0.75pt]  [font=\scriptsize]  {$0$};
	% Text Node
	\draw (94,545.09) node [anchor=north west][inner sep=0.75pt]  [font=\scriptsize]  {$0$};
	% Text Node
	\draw (114,545.09) node [anchor=north west][inner sep=0.75pt]  [font=\scriptsize]  {$0$};
	% Text Node
	\draw (134.25,545.09) node [anchor=north west][inner sep=0.75pt]  [font=\scriptsize]  {$0$};
	% Text Node
	\draw (154,545.09) node [anchor=north west][inner sep=0.75pt]  [font=\scriptsize]  {$1$};
	% Text Node
	\draw (174,545.09) node [anchor=north west][inner sep=0.75pt]  [font=\scriptsize]  {$0$};
	% Text Node
	\draw (194.25,545.09) node [anchor=north west][inner sep=0.75pt]  [font=\scriptsize]  {$0$};
	% Text Node
	\draw (214,545.09) node [anchor=north west][inner sep=0.75pt]  [font=\scriptsize]  {$1$};
	% Text Node
	\draw (254,545.09) node [anchor=north west][inner sep=0.75pt]  [font=\scriptsize]  {$1$};
	% Text Node
	\draw (273.75,545.09) node [anchor=north west][inner sep=0.75pt]  [font=\scriptsize]  {$0$};
	% Text Node
	\draw (294,545.09) node [anchor=north west][inner sep=0.75pt]  [font=\scriptsize]  {$0$};
	% Text Node
	\draw (234.25,545.09) node [anchor=north west][inner sep=0.75pt]  [font=\scriptsize]  {$1$};
	% Text Node
	\draw (313.75,545.09) node [anchor=north west][inner sep=0.75pt]  [font=\scriptsize]  {$1$};
	% Text Node
	\draw (334,545.09) node [anchor=north west][inner sep=0.75pt]  [font=\scriptsize]  {$0$};
	% Text Node
	\draw (60,452.59) node [anchor=north west][inner sep=0.75pt]    {$X$};
	% Text Node
	\draw (60,542.84) node [anchor=north west][inner sep=0.75pt]    {$Y$};
	% Text Node
	\draw (159.14,395.77) node [anchor=north west][inner sep=0.75pt]  [font=\footnotesize]  {$\alpha =1 \Rightarrow \text{\#count=} 4$};

\end{tikzpicture}

%% file: src/offline.tex
\section{Dynamic Subset  Sum with Offline Access}\label{sec:off}
In this section, we consider the dynamic subset sum problem when offline access to the operations is available.
For a time step $1 \leq i \leq \opers$, let $\sums_i \subseteq [0,\tmax] $ be a set of positive integers where $\alpha \in \sums_i$ if and only if there exists a subset $S$ of items added to the problem in the first $i$ operations such that $\sum_{j \in S}{w_j} = \alpha$. In our algorithm, we compute $\sums_i$ for some of the time steps before starting to process the queries and use this additional information to improve the amortized processing time of the operations.

It follows that $\sums_i$s are monotone in that $\sums_i \subseteq \sums_{i+1}$ holds for all $1 \leq i < \opers$. Moreover, $\sum |\sums_{i+1} \setminus \sums_i| \leq \tmax$ holds by definition. Our algorithm for dynamic subset sum with offline access takes overall time $\tilde O(\opers+\tmax \min{\{\sqrt{\opers},\sqrt{\tmax}\}})$ which is $\tilde O(1+\tmax \min{\{\sqrt{\opers},\sqrt{\tmax}\}}/\opers)$ per operation. 

%We call an algorithm \emph{output-sensitive} if we only if it perform an $O(f(\tmax, \opers))$ computation every time a member is added to $\sums$. Such an output-sensitive algorithm has an $O(\tmax f(\tmax, \opers)/\opers)$ amortized processing time since every $j \in [\tmax]$ is added only once. So, we focus on finding the changes to $\sums$ after each operation.
%So, we focus on finding $C_i \setminus C_{i-1}$ after the $i$-th query, keeping in mind that we can spend a time dependent on the size of $C_i \setminus C_{i-1}$.

\SetKw{Continue}{continue}
\SetKw{Break}{break}
\SetKw{Return}{return}

\begin{algorithm}[ht]
	\SetAlgoLined
	\KwIn{The maximum target threshold $\tmax$, a sequence of $\opers$ operations of following form: either (1) addition of an item $w_j$, or (2) query $t_j$.}
	%$\textsf{next-milestone} \gets 1$\;
	%$\sums_1,  X \gets \emptyset$\;
	$\important, \mathcal{Q} \gets \textsf{FilterOperations}(\opers,\tmax)$\;
	%$m_1, m_2, \ldots, m_{\milestones^*} \gets \textsf{FindMilestones}(\important)$\;	
	$m_0 \gets 0$\;
	%\tcp{$\quad\textsf{Bringmann}(\{w_l \mid l \in \important \cap [m']\})\setminus\textsf{Bringmann}(\{w_l \mid l \in \important \cap [m]\}).$}
	\For{$i \in [\lambda]$}{
		\tcp{We find $\sums_j$ whenever needed by calling $\textsf{Bringmann}(\{w_l \mid l \in \important \cap [j]\})$}
		Binary search on variable $m_i$ comparing $|\sums_{m_i}| \geq i\tmax/\milestones$\;%\quad\forall{i \in [\milestones^*]}$\;
		Let $X_{i} \gets  \sums_{m_{i}-1} \setminus \sums_{m_{i-1}}$\;%\quad\forall{i \in [\milestones^*]}$\;
	}
$i \gets 1$\;
	$A \gets (1, 0, \ldots, 0) \in \{0,1\}^{\tmax}$\;
	\For{$j \in \important \cup \mathcal{Q}$ \emph{sorted by} $j$}{
		\eIf{$j \in \mathcal{Q}$}{
			Answer $\textsf{Yes}$ if $t_j \in A$ and $\textsf{No}$ otherwise\;
		}{
			\If{$j \in \{ m_i, m_i-1 \mid i \in [\lambda]\}$}{
				$i = \max{\{ i \mid  i \in [\lambda] \textsf{ and } 0 \leq m_{i-1} - j \leq 1  \}} + 1$\;
				$A \gets \sums_j$\;
				%$i \gets i + 1$\;
				%$X \gets \textsf{Bringmann}(\{w_l \mid l \in \important \cap [m_i-1]\}) \setminus \sums$\;
				%$\textsf{next-milestone}, X \gets \textsf{NextMilestone}(\important)$\;
			}
		%	Add $x$ to $A$ if $x - w_j \in A \textsf{ and } x \not\in A\quad\forall{x \in X_i}$\;
			\For {$x \in X_i$}{
				\If{$x - w_j \in A \textsf{\emph{ and }} x \not\in A$}{
					Add $x$ to $A$\;
				% and remove $x$ from $X$\;
				}
			}	
		}
	%$\sums_{j+1} = \sums_j$\;
	}		
	\caption{Dynamic Subset Sum with Offline Access $O(1+\tmax\min{\{\sqrt{\opers}, \sqrt{\tmax}\}}/\opers)$ amortized processing time}\label{alg:offline}
\end{algorithm}

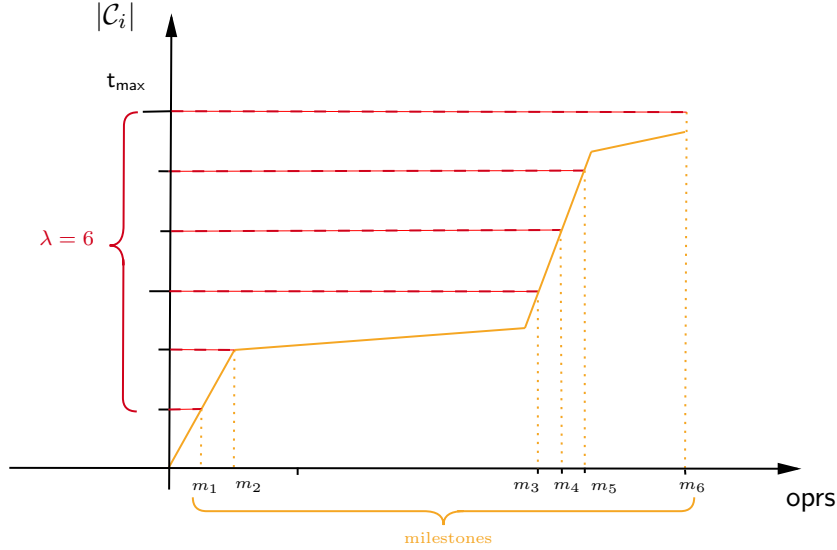
\begin{figure}[ht]
	\centering
	\input{figs/offline-milestones}
	\caption{The milestones are illustrated with an example in which $\milestones=6$.}\label{fig:milestone}
\end{figure}

\iffalse

\begin{algorithm}[ht]
	\SetAlgoLined
	\KwIn{The set of important addition operations $\important$.}%, the index of current operation $j$, and the set $\sums$ of all possible sums using first $j$ operations.}
	\tcp{Let $X_{m, m'}$ show the following set}
	\tcp{$\quad\textsf{Bringmann}(\{w_l \mid l \in \important \cap [m']\})\setminus\textsf{Bringmann}(\{w_l \mid l \in \important \cap [m]\}).$}
	$m_0 \gets 0$\;
	Binary search on variable $m_i$ comparing $|X_{m_{i-1},m_i}| \geq \tmax/\milestones^*$ for all $i \in [\milestones^*]$\;
	%\For{$i \gets [\milestones^*]$} {
		%$X \gets \textsf{Bringmann}(\{w_l \mid l \in \important \cap [m_i]\}) \setminus \sums$\;
		%$X \gets \sums_{\textsf{next}} \setminus \sums$\;
		%\If{$|X_{m_{i-1},m_{i}}| > 10\tmax/\milestones^*$}{
		%	\If{$m_i - m_{i-1} \geq 2 $}{
				%$X \gets \textsf{Bringmann}(\{w_l \mid l \in \important \cap [m_i-1]\}) \setminus \sums$\;	
		%		$m_i  \gets m_i-1$\;				
		%	}
		%}
	%}
	\Return $m_1, m_2, \ldots, m_{\milestones^*}$\;
	\caption{$\textsf{FindMilestones}(\important)$\\Finding the set of \emph{milestones} given the sequence of operations.}
	% \\$O(\tmax/\opers+\sqrt{\tmax})$ amortized processing time}
	\label{alg:milestone}
\end{algorithm}

\fi

\begin{algorithm}[ht]
	\SetAlgoLined
	\KwIn{The maximum target threshold $\tmax$, a sequence of $\opers$ operations of following form: either (1) addition of an item $w_j$, or (2) query $t_j$.}
	$\important, \mathcal{Q} \gets \emptyset$\;
	$\textsf{cnt}_j \gets 0 \quad\forall\;j\in[\tmax]$\;
	%$m_1, m_2, \ldots, m_k \gets \textsf{FindMilestones}()$\;
	%$C_{m_i} \gets \textsf{SubsetSum}(\{z_1, z_2, \ldots, z_{m_i}\}, \tmax)\quad\forall{i \in [k]}$\;
	
	\For{$j \in [\opers]$}{
		\eIf{operation $j$ is addition of an item $w_j$}{
			$\textsf{cnt}_{w_j} \gets \textsf{cnt}_{w_j} + 1$\;
			\If{$\textsf{cnt}_{w_j} \leq \lfloor\tmax / w_j\rfloor$}{$\important \gets \important \cup \{j\}$\;}
		}{
			$\mathcal{Q} \gets \mathcal{Q} \cup \{ j \}$\;
		}
	}
	\Return $\important, \mathcal{Q}$\;
	\caption{$\textsf{FilterOperations}(\opers,\tmax)$}\label{alg:important}
\end{algorithm}

We divide the operations by $\milestones$ \textit{milestones} in our algorithm. Let the milestones be at operations $m_1, m_2, \ldots, m_{\lambda}$. The time steps between consecutive milestones may vary depending on how the solution changes. More precisely, the $i$'th milestone is at the smallest operation $j \in [\opers]$ such that $|\sums_j| \geq i \tmax/\milestones$ (see Figure~\ref{fig:milestone} for an example). It is straightforward to find the set of milestones $m_1, m_2, \ldots, m_\milestones$ by using the Bringmann's algorithm as a blackbox in a binary search. For finding $m_i$, we need to decide for an index $j$ whether $j < m_i$ or $j \geq m_i$ at each step of a binary search. To this end, we solve the subset sum problem from scratch for the first $j$ operations in $\widetilde{O}(j+\tmax)$ time and compare $|\sums_j|$ with $i \tmax/\milestones$. If we repeat this process $\milestones$ times for every $i \in [\milestones]$, we invoke the Bringmann's algorithm at most $O(\milestones\log{\opers})$ times, and therefore we spend a total of $\widetilde{O}(\milestones(\opers+\tmax))$ running time for finding the set of milestones. Although~\cite{bringmann2017near} only discusses a YES/NO variant of subset sum for a particular target value $t$, it is easy to see that the same algorithm can be used to solve the problem for all target values in range $[0,t]$. A formal discussion of the ideas is given in~\cite{bateni2018fast}.

We assume for now that $\opers \leq \tmax$ and later modify our algorithm to also work for $\opers > \tmax$. In our algorithm, before we process any of the operations, for each milestone $m_i$, we compute $\sums_{m_i}$ and $\sums_{m_i-1}$ from scratch. This takes time  $\tilde {O}(\opers+\tmax)$ for each milestone using Bringmann's algorithm~\cite{bringmann2017near}. After that, we make an array $A[0,\tmax] \rightarrow \{0,1\}$ such that $A_0 = 1$ and $A_i = 0 $ for all $1 \leq i \leq \tmax$ which represents the solution. We then process the operations one by one. If an operation adds an element we update $A$ and otherwise we report the solution by looking at $A$.

When we reach a milestone, say $m_i$, updating $A$ would be trivial; we just look at $C_{m_i}$ which is precomputed in the beginning and based on that we update $A$. For operations between two consecutive milestones, say $m_i$ and $m_{i+1}$, we know that $A$ only changes for indices in $\sums_{m_{i+1}-1} \setminus \sums_{m_{i}}$. Since we already computed $\sums_{m_i}$ and $\sums_{m_{i+1}-1}$,  we have explicit access to $\sums_{m_{i+1}-1} \setminus \sums_{m_i}$. Let $X_i = \sums_{m_{i}-1} \setminus \sums_{m_{i-1}}$. It follows from the definition that $|X_i| \leq \tmax/\milestones$.

In order to process operations between $m_i$ and $m_{i+1}-1$, we only consider set $X_i$. More precisely, after each operation that adds an item with weight $w$ to the problem, we only consider elements of $A$ that are included in $X_i$. For an integer $x \in X_i$, we set $A_x = 1$ if $A_{x-w}$ was equal to $1$ prior to this operation. This way, we can maintain $A$ after each operation by iterating over the elements of $X_i$ and thus spending $O(|X_i|) = O(\tmax/\milestones)$ time. By maintaining $A$ after each insertion operation, we can answer each query operation in time $O(1)$ by verifying if the value of $A$ is equal to $1$ for the given target. 

The overall time that our algorithm spends in order to process the operations is as follows: In the beginning, we first find the positions for the milestones. This takes a total time of $\tilde O(\lambda(\opers + \tmax))$. We then compute $\sums_i$ for $i \in \{m_1, m_1-1, m_2, m_2-1, \ldots, m_\lambda, m_{\lambda}-1\}$ each in time $\tilde O(\opers + \tmax)$. This amount to a total runtime of $\tilde O(\milestones (\opers + \tmax))$. From then on, each time we update $A$ for operations that are in the middle of two milestones, we only spend time $O(\tmax/\lambda)$. Also, each time we reach a milestone $m_i$, we update $A$ by using the already computed $\sums_{m_i}$ in time $O(\tmax)$. This also takes a total time of $O(\lambda \tmax)$. Thus, the total runtime of the algorithm is $\tilde O(\lambda(\opers + \tmax) + \opers \cdot\tmax /\lambda )$. Thus, per operation we spend time $\tilde O( \lambda \tmax / \opers + \tmax/\lambda)$. When $\opers \leq \tmax$, we set $\lambda = \sqrt{\opers}$ which gives us an amortized processing time of $\tilde O(\tmax / \sqrt{\opers})$.

If $\opers > \tmax$ we make a distinction between insertion operations and query operations. Notice that in our algorithm, the processing time of each query operation is $O(1)$. With a similar argument that we made in Section~\ref{sec:dynamic}, we can ignore some redundant insertion operations to make sure their count is bounded by $\tilde O(\tmax)$ without any change to the outcome of the algorithm. Therefore, we can improve the initialization time of our algorithm down to $\tilde O(\lambda \tmax)$ by only considering the insertion operations that are not ignored. Similarly, updating $A$ for the insertion operations between the milestones would take time $O(\tmax/\lambda)$. Moreover, answering each query operation takes time $O(1)$ and therefore the total time we spend on all operations combined is $\tilde O(\opers + \lambda \tmax + \tmax^2 / \lambda)$. By setting $\lambda = \sqrt{\tmax}$ we obtain amortized processing time $\tilde O(1 + \tmax^{1.5}/\opers)$.

\vspace{5mm}
\begin{theorem}\label{thm:offline}
	There exists a randomized algorithm for dynamic subset sum with offline access whose amortized processing time is bounded by  $\tilde O(1+\tmax \min{\{\sqrt{\opers},\sqrt{\tmax}\}}/\opers)$. The algorithm succeeds with high probability.
\end{theorem}

%% file: figs/offline-milestones.tex
\tikzset{every picture/.style={line width=0.75pt}} %set default line width to 0.75pt        

\begin{tikzpicture}[x=0.75pt,y=0.75pt,yscale=-1,xscale=1]
	%uncomment if require: \path (0,662); %set diagram left start at 0, and has height of 662
	
	%Straight Lines [id:da1804517207891252] 
	\draw [color={rgb, 255:red, 245; green, 166; blue, 35 }  ,draw opacity=1 ] [dash pattern={on 0.84pt off 2.51pt}]  (222.3,210.14) -- (221.98,269.84) ;
	%Straight Lines [id:da45786966641921567] 
	\draw [color={rgb, 255:red, 208; green, 2; blue, 27 }  ,draw opacity=1 ][fill={rgb, 255:red, 255; green, 0; blue, 0 }  ,fill opacity=1 ] [dash pattern={on 4.5pt off 4.5pt}]  (189.07,119.96) -- (398.97,119.59) ;
	%Straight Lines [id:da5013604796457034] 
	\draw [color={rgb, 255:red, 245; green, 166; blue, 35 }  ,draw opacity=1 ] [dash pattern={on 0.84pt off 2.51pt}]  (450.36,90.09) -- (449.54,270.59) ;
	%Straight Lines [id:da6868351574871974] 
	\draw [color={rgb, 255:red, 245; green, 166; blue, 35 }  ,draw opacity=1 ]   (189.03,269.74) -- (222.3,210.14) ;
	%Straight Lines [id:da31977965221169913] 
	\draw    (108.76,269.64) -- (506.27,270.09) ;
	\draw [shift={(508.27,270.09)}, rotate = 180.06] [fill={rgb, 255:red, 0; green, 0; blue, 0 }  ][line width=0.08]  [draw opacity=0] (12,-3) -- (0,0) -- (12,3) -- cycle    ;
	%Straight Lines [id:da7091936608091394] 
	\draw    (189.35,280.59) -- (190.51,42.2) ;
	\draw [shift={(190.52,40.2)}, rotate = 450.28] [fill={rgb, 255:red, 0; green, 0; blue, 0 }  ][line width=0.08]  [draw opacity=0] (12,-3) -- (0,0) -- (12,3) -- cycle    ;
	%Straight Lines [id:da7504399820828549] 
	\draw    (398.97,269.59) -- (399,272.27) ;
	%Straight Lines [id:da4411592394743922] 
	\draw    (449.54,270.59) -- (449.71,273.04) ;
	%Straight Lines [id:da5273189801809726] 
	\draw    (189.07,119.96) -- (183.99,119.92) ;
	%Straight Lines [id:da784465713499481] 
	\draw    (189.58,180.49) -- (179.33,180.49) ;
	%Straight Lines [id:da7244652501844011] 
	\draw    (183.99,240.2) -- (189.12,240.2) ;
	%Straight Lines [id:da4991084863046744] 
	\draw    (254.17,269.49) -- (254.17,273.21) ;
	%Straight Lines [id:da033748997346433995] 
	\draw [color={rgb, 255:red, 208; green, 2; blue, 27 }  ,draw opacity=1 ][fill={rgb, 255:red, 255; green, 0; blue, 0 }  ,fill opacity=1 ] [dash pattern={on 4.5pt off 4.5pt}]  (190.52,89.63) -- (450.36,90.09) ;
	%Straight Lines [id:da0628792753937839] 
	\draw    (190.05,209.92) -- (183.99,209.92) ;
	%Straight Lines [id:da052178626468537104] 
	\draw    (190.52,89.63) -- (176.07,89.92) ;
	%Straight Lines [id:da14885596623423702] 
	\draw    (190.52,150.2) -- (184.92,150.2) ;
	%Straight Lines [id:da9740942178245888] 
	\draw [color={rgb, 255:red, 245; green, 166; blue, 35 }  ,draw opacity=1 ]   (222.3,210.14) -- (368.79,199.09) ;
	%Straight Lines [id:da663518629989301] 
	\draw [color={rgb, 255:red, 245; green, 166; blue, 35 }  ,draw opacity=1 ]   (368.79,199.09) -- (402.24,110.09) ;
	%Straight Lines [id:da137350072058829] 
	\draw [color={rgb, 255:red, 245; green, 166; blue, 35 }  ,draw opacity=1 ]   (402.24,110.09) -- (449.54,100.09) ;
	%Straight Lines [id:da5748890298649156] 
	\draw [color={rgb, 255:red, 208; green, 2; blue, 27 }  ,draw opacity=1 ][fill={rgb, 255:red, 255; green, 0; blue, 0 }  ,fill opacity=1 ] [dash pattern={on 4.5pt off 4.5pt}]  (190.52,150.2) -- (386.74,149.59) ;
	%Straight Lines [id:da504226139220135] 
	\draw [color={rgb, 255:red, 208; green, 2; blue, 27 }  ,draw opacity=1 ][fill={rgb, 255:red, 255; green, 0; blue, 0 }  ,fill opacity=1 ] [dash pattern={on 4.5pt off 4.5pt}]  (189.58,180.49) -- (375.32,180.59) ;
	%Straight Lines [id:da9213142871417144] 
	\draw [color={rgb, 255:red, 208; green, 2; blue, 27 }  ,draw opacity=1 ][fill={rgb, 255:red, 255; green, 0; blue, 0 }  ,fill opacity=1 ] [dash pattern={on 4.5pt off 4.5pt}]  (190.05,209.92) -- (222.3,210.14) ;
	%Straight Lines [id:da5863859891458614] 
	\draw [color={rgb, 255:red, 208; green, 2; blue, 27 }  ,draw opacity=1 ][fill={rgb, 255:red, 255; green, 0; blue, 0 }  ,fill opacity=1 ] [dash pattern={on 4.5pt off 4.5pt}]  (189.12,240.2) -- (205.66,239.94) ;
	%Straight Lines [id:da7454650059154755] 
	\draw [color={rgb, 255:red, 245; green, 166; blue, 35 }  ,draw opacity=1 ] [dash pattern={on 0.84pt off 2.51pt}]  (398.97,119.59) -- (398.97,269.59) ;
	%Straight Lines [id:da028818059955194375] 
	\draw [color={rgb, 255:red, 245; green, 166; blue, 35 }  ,draw opacity=1 ] [dash pattern={on 0.84pt off 2.51pt}]  (386.74,149.59) -- (387.4,269.74) ;
	%Straight Lines [id:da4050912832770901] 
	\draw [color={rgb, 255:red, 245; green, 166; blue, 35 }  ,draw opacity=1 ] [dash pattern={on 0.84pt off 2.51pt}]  (375.32,180.59) -- (375.32,270.09) ;
	%Straight Lines [id:da6274576464556803] 
	\draw    (387.4,269.74) -- (387.46,272.57) ;
	%Straight Lines [id:da13651477131468126] 
	\draw    (375.32,270.09) -- (375.32,272.64) ;
	%Straight Lines [id:da4318746784912504] 
	\draw [color={rgb, 255:red, 245; green, 166; blue, 35 }  ,draw opacity=1 ] [dash pattern={on 0.84pt off 2.51pt}]  (205.66,239.94) -- (205.34,269.64) ;
	%Shape: Brace [id:dp3224925521200428] 
	\draw  [color={rgb, 255:red, 245; green, 166; blue, 35 }  ,draw opacity=1 ] (201.31,284.51) .. controls (201.31,289.18) and (203.64,291.51) .. (208.31,291.51) -- (317.74,291.51) .. controls (324.41,291.51) and (327.74,293.84) .. (327.74,298.51) .. controls (327.74,293.84) and (331.07,291.51) .. (337.74,291.51)(334.74,291.51) -- (447.17,291.51) .. controls (451.84,291.51) and (454.17,289.18) .. (454.17,284.51) ;
	%Shape: Brace [id:dp6980942980200058] 
	\draw  [color={rgb, 255:red, 208; green, 2; blue, 27 }  ,draw opacity=1 ] (173.58,90.18) .. controls (168.91,90.16) and (166.57,92.48) .. (166.56,97.15) -- (166.37,147.42) .. controls (166.35,154.09) and (164.01,157.41) .. (159.34,157.39) .. controls (164.01,157.41) and (166.33,160.75) .. (166.3,167.42)(166.31,164.42) -- (166.05,234.15) .. controls (166.04,238.82) and (168.36,241.16) .. (173.03,241.18) ;
	
	% Text Node
	\draw (156.46,69.13) node [anchor=north west][inner sep=0.75pt]  [font=\footnotesize]  {$\tmax$};
	% Text Node
	\draw (199.23,275.69) node [anchor=north west][inner sep=0.75pt]  [font=\tiny]  {$m_{1}$};
	% Text Node
	\draw (150.6,32.92) node [anchor=north west][inner sep=0.75pt]    {$| \mathcal{C}_i |$};
	% Text Node
	\draw (221.09,275.09) node [anchor=north west][inner sep=0.75pt]  [font=\tiny]  {$m_{2}$};
	% Text Node
	\draw (361.39,275.29) node [anchor=north west][inner sep=0.75pt]  [font=\tiny]  {$m_{3}$};
	% Text Node
	\draw (381.68,275.24) node [anchor=north west][inner sep=0.75pt]  [font=\tiny]  {$m_{4}$};
	% Text Node
	\draw (400.54,275.69) node [anchor=north west][inner sep=0.75pt]  [font=\tiny]  {$m_{5}$};
	% Text Node
	\draw (444.91,275.29) node [anchor=north west][inner sep=0.75pt]  [font=\tiny]  {$m_{6}$};
	% Text Node
	\draw (499.08,280.92) node [anchor=north west][inner sep=0.75pt]    {$\textsf{oprs}$};
	% Text Node
	\draw (306.47,300.68) node [anchor=north west][inner sep=0.75pt]   [align=left] {{\tiny \textcolor[rgb]{0.96,0.65,0.14}{milestones}}};
	% Text Node
	\draw (122.45,148.01) node [anchor=north west][inner sep=0.75pt]  [font=\scriptsize,color={rgb, 255:red, 208; green, 2; blue, 27 }  ,opacity=1 ]  {$\lambda=6$};

\end{tikzpicture}

%% file: src/remove-hardness.tex
\section{Hardness of Fully Dynamic Subset Sum}\label{sec:hardness}
%\Saeed{This section is informal and short. Please make it more formal by defining everything clearly and adding theorems and proofs and pseudocodes. Make sure it is not shorter than 2 pages.} \Hamed{Explain the hardness. "Inspired by folani"}
%In this section,  %\Saeed{Please update the notation: e.g. $T \rightarrow \tmax$}
In this section, we aim to show a lower bound for the fully dynamic subset sum problem. Recall that in this problem, we also allow for deletion operations. We prove that there is no fully dynamic subset sum algorithm with amortized processing time $O(\tmax^{1-\epsilon})$ for any constant $\epsilon > 0$, unless  \textsf{SETH}  fails. Our result is mainly inspired by the  \textsf{SETH}-based hardness of~\cite{abboud2019seth} for subset sum. Our proof directly makes use of some of the observations made by~\cite{abboud2019seth} as blackbox. We  begin by explaining their method in Section~\ref{sec:abboud} and then proceed by showing how it gives a lower bound for the fully dynamic subset sum problem in Section~\ref{sec:hamed}. Our lower bound is formally stated in Theorem~\ref{thm:hardness}.

\begin{theorem}\label{thm:hardness}
	For any arbitrarily large constant $c > 0$ and any arbitrarily small constant $\epsilon > 0$, there is no fully dynamic algorithm for subset sum with amortized processing time $O(\tmax^{1-\epsilon})$ when the number of operations is $\Omega(\tmax^c)$ unless  \textsf{SETH}  fails.
\end{theorem}

In our proof, we build a collection of $N = \Theta(\tmax^c)$ instances of the subset sum problem $\inst_1, \inst_2, \ldots, \inst_N$ each with $O(\log \tmax)$ items and target $t_i$ which is bounded by $\tmax$. 
We denote by $\inst_i^+$ a set of operations that adds the items of $\inst_i$ to the dynamic problem, and by $\inst_i^-$ a set of operations that removes the items of $\inst_i$ from the dynamic problem. In addition, let $\inst_i^\textsf{?}$ be a query operation that asks if we can build a sum of $t_i$ using the elements of the dynamic problem at the time of the query operation. Then we consider the following sequence of operations for the dynamic subset sum problem which has a combination of insertion and deletion, and query operations:

\[ \inst_1^+, \inst_1^\textsf{?}, \inst_1^-, \inst_2^+, \inst_2^\textsf{?}, \inst_2^-, \ldots, \inst_N^+, \inst_N^\textsf{?}, \inst_N^- \] 

Roughly speaking, we show that under the  \textsf{SETH}  hypothesis, there is no algorithm that is able to find out whether the result of any of $\inst_i$'s is Yes in time $O(N \tmax^{1-\epsilon})$  for any constant $\epsilon > 0$. 
%Note that we set $\gamma = 1$ for simplicity. 
However if we feed this sequence into a fully dynamic subset sum algorithm with amortized processing time $O(\tmax^{1-\epsilon})$, it will solve this problem in time $\tilde O(N \tmax^{1-\epsilon})$ which contradicts \textsf{SETH}.

\subsection{The Reduction of~\cite{abboud2019seth}}\label{sec:abboud}

\cite{abboud2019seth} shows a clever reduction from $k$-SAT to subset sum. Let $\phi$ be an instance of $k$-SAT on $m$ variables, a CNF formula in which every clause has at most $k$ variables. It is well-known that finding an assignment of true or false to these $m$ variables that satisfies $\phi$ is NP-hard if $k \geq 3$ and it can be solved in $O_m(2^m)$\footnote{$O_m$ notation hides polynomials in terms of $m$.} time by verifying every possible assignment. The Strong Exponential Time Hypothesis~\cite{calabro2009complexity,impagliazzo2001complexity}, or  \textsf{SETH}  for short, states that for any constant $\epsilon > 0$, there is a constant $k \geq 3$ so that $k$-SAT is not solvable in time $O(2^{(1-\epsilon)m})$. In their main theorem, they use the reduction  from $k$-SAT to subset sum to show that the existence of an algorithm that solves subset sum in time $t^{1-\epsilon}2^{o(n)}$ refutes  \textsf{SETH}, as it can also solve $k$-SAT in time $O(2^{(1-\epsilon/5)m})$.

Their approach utilizes several techniques  including a sparsification lemma of~\cite{impagliazzo2001problems} and $\lambda$-average-free sets of Behrend~\cite{behrend1946sets} to achieve the following reduction which we use as a black-box in this paper.

\begin{lemma}[\cite{abboud2019seth}]\label{lemma:abud}
	\label{lemma:reduction}
	For any constant $\epsilon > 0$ and any constant $k \geq 3$,  a $k$-SAT formulation $\phi$ on $m$ variables can be reduced to $l \leq 2^{\epsilon m}$ instances of subset sum $\inst_1, \inst_2, \ldots, \inst_l$, in a way that $\phi$ is satisfiable if and only if at least one of $\inst_i$'s returns Yes. In each $\inst_i$, target $t_i$ is bounded by $2^{(1+2\epsilon)m}$ and the number of items is at most $f(k,\epsilon) m$, where $f$ is a function that only depends on $k$ and $\epsilon$.
\end{lemma}

Lemma~\ref{lemma:abud} directly gives a lower bound to subset sum in the following way. Assume for the sake of contradiction that subset sum can be solved in time $t^{1-\epsilon'}2^{o(n)}$ and let $\epsilon = \epsilon'/4$. We can then use Lemma~\ref{lemma:abud} to solve each of the $l$ instances of subset sum in time 
\begin{equation*}
\begin{split}
O(\big(2^{(1+2\epsilon)m} \big)^{1-\epsilon'}2^{o(f(k,\epsilon) m)}) &= O(\big(2^{(1+2\epsilon)m} \big)^{1-\epsilon'}2^{o(m)})\\
& = O(2^{(1+2\epsilon)(1-\epsilon')m}2^{o(m)})\\
& = O(2^{(1+2\epsilon)(1-4\epsilon)m}2^{o(m)})\\
& = O(2^{(1-2\epsilon)m})\\
\end{split}
\end{equation*}
which amounts to a total runtime of $O(2^{(1-\epsilon)m})$ for all $l \leq 2^{\epsilon m}$  instances. This contradicts \textsf{SETH} for sufficiently large $k$.

\subsection{A Lower Bound for Amortized Processing Time of Fully Dynamic Subset Sum}\label{sec:hamed}

In this section we prove Theorem~\ref{thm:hardness}.

\begin{proof}[Of Theorem~\ref{thm:hardness}]
As mentioned earlier, we set $N = \tmax^c$ and seek to find $N$ instances of subset sum $\inst_1, \inst_2, \ldots, \inst_N$ in a way that the solution to some $k$-SAT problem is Yes if and only if the solution for one of the subset sum problems $\inst_1, \inst_2, \ldots, \inst_N$ is Yes.

To this end, we assume for the sake of contradiction that for some $\epsilon'$, there exists a fully dynamic subset sum algorithm with amortized processing time $t^{1-\epsilon'}$. Let $\epsilon = \epsilon'/3$ and set $m$  in a way that $2^{(1+2\epsilon) m} = \tmax$. That is, $m = \frac{\log \tmax}{1+2\epsilon}$. We assume that $N \geq 2^{\epsilon m}$ otherwise the proof is similar to the lower bound for subset sum. We also set $m'$ in a way that $2^{m'-m} = N/2^{\epsilon m}$. Therefore, $m' = m + \log N-\epsilon m$.

We argue that any $k$-SAT instance on $m'$ variables can be reduced to at most $N$ instances of subset sum. To this end, we divide the variables into two sets. The first set contains the first $m$ variables and the second set contains the last $m'-m$ variables. We try all the $2^{m'-m}$ different assignments for the second set of variables, and for each one we make a separate $k$-SAT instance with $m$ variables which is subject to the first $m$ variables and the given assignment to the second set of variables. Obviously, at least one of the $2^{m'-m}$ instances of $k$-SAT is solvable if and only if we can solve the original $k$-SAT problem on $m'$ variables. We then reduce each of the $k$-SAT instance on $m$ variables into $l \leq 2^{\epsilon m}$ instances of subset sum each with target value $t_i$ bounded by $2^{(1+2\epsilon) m}$ and at most $f(k,\epsilon)m$ items.

Let the subset sum instances be $\inst_1, \inst_2, \ldots, \inst_N$ and the target for each $\inst_i$ be $t_i$. If the solution to any of the subset sum instances is Yes, the original $k$-SAT problem admits a solution.  We denote by $\inst_i^+$ a set of operations that adds the items of $\inst_i$ to the dynamic problem and by $\inst_i^-$ a set of operations that removes the items of $\inst_i$ from the dynamic problem. In addition, let $\inst_i^\textsf{?}$ be a query operation that asks if we can build a sum of $t_i$ using the elements of the dynamic problem at the time of the query operation. Then we consider the following sequence of operations for the dynamic subset sum problem which has a combination of insertion and deletion, and query operations:

\[ \inst_1^+, \inst_1^\textsf{?}, \inst_1^-, \inst_2^+, \inst_2^\textsf{?}, \inst_2^-, \ldots, \inst_N^+, \inst_N^\textsf{?}, \inst_N^- \] 

Since $t_i$ for each instance is bounded by $\tmax = 2^{(1+2\epsilon) m}$ and each instance has at most $f(\epsilon, k)m = O(\log \tmax)$ elements, using the fully dynamic subset sum algorithm, we can find out if the solution to any of the subset sum instance is Yes in time $\tilde O(N\tmax^{1-\epsilon'}) = \tilde O(\tmax^{c+1-\epsilon'})$. This implies that the original subset sum problem on $m'$ variables can be solved in time $\tilde O(\tmax^{c+1-\epsilon'}) = \tilde O(2^{\log \tmax(c+1-\epsilon')})$. Since 
\begin{equation*}
\begin{split}
m'  &= m + \log N-\epsilon m\\
& = m(1-\epsilon) + c \log \tmax \\
&= (1-\epsilon)\frac{\log \tmax}{1+2\epsilon}+ c \log \tmax \\
&= \log \tmax (c+\frac{1-\epsilon}{1+2\epsilon}) \\
&\geq \log \tmax (c+1-3\epsilon) \\
&= \log \tmax (c+1-\epsilon')
\end{split}
\end{equation*}
this contradicts \textsf{SETH} for large enough $k$.
\end{proof}

%~\cite{abboud2019seth} builds this collection of instances starting with a k-SAT instance of size $n = O(\log T)$. They use the self-reduction property of k-SAT to reduce the initial k-SAT instance to $2^{n_1}$ instances of k-SAT each with size $n - n_1$, one for each $2^{n_1}$ possible assignment of values to the first $n_1$ bits in the initial k-SAT instance. Solving either of these instances affirmatively means that the answer to the original k-SAT instance is also affirmative. Thus, mixing this with their main contribution which is a reduction from k-SAT instance to a number of subset sum instances, we conclude that solving dynamic subset sum allowing remove operations in $O(T^{1-\epsilon})$ amortized processing time leads to an $O(2^{(1-\epsilon)n})$ algorithm for k-SAT which refutes the  \textsf{SETH}  hypothesis.

%\begin{claim}\label{hard-instance}
%	There is a family of hard instances of the subset sum problem with size at least $\Omega(n)$ where solving any subset of these instances does not help with solving any instance outside this subset.
%\end{claim}

%% file: src/3sum.tex
\section{Fully Dynamic  Bounded $3$-Sum and Beyond}\label{sec:3sum}
In the fully dynamic bounded $3$-sum problem, we have three sets $\Aa$, $\Bb$, and $\Cc$ which are initially empty. These sets will contain numbers with values in range $[0, \rmax]$, where $\rmax$ is an upper bound explicitly given to the algorithm in advance. At each point in time, one of the following operations arrives: 
\begin{itemize}
	\item Insert a number with value $w \in [0,\rmax]$ into $\Aa$, $\Bb$, or $\Cc$.
	\item Delete an existing number with value $w \in [0,\rmax]$ from $\Aa$, $\Bb$, or $\Cc$.
	\item Report whether three numbers $(a,b,c) \in \Aa \times \Bb \times \Cc$ exist so that $a + b = c$.
\end{itemize}
We assume for simplicity that a number which already exists in a set will never be added to that set again and that each deletion operation is valid in the sense that the number to be deleted is previously added to the corresponding set. We give a fully dynamic algorithm for  bounded $3$-sum with amortized processing time %$O(1+\rmax\cdot\min{\{\sqrt{\opers}, \sqrt{\rmax}\}}/\opers)$ 
$\tilde O(\min{\{\opers, \rmax^{0.5}\}})$ in Theorem~\ref{thm:3sum}.
%For the special case of dynamic bounded $3$-sum in which there are no \emph{deletion} operations we improve the amortized processing time to $\tilde O(1 + \min{\{\opers, \rmax^{0.5},\rmax^{1.5}/\opers\}})$ in Corollary~\ref{thm:3sum-addition}.
%, matching the complexity of dynamic subset sum with offline access (Theorem~\ref{thm:offline}). 
We also extend our solution to the  bounded $k$-sum problem, in which there are $k$ sets $\Aa_1, \Aa_2, \ldots, \Aa_k$ and we wish to determine if there exists a tuple $(a_1, a_2, \ldots, a_k)$ such that $a_1+a_2+\ldots+a_{k-1}=a_k$ and $a_i \in \Aa_i$. Next, in Lemma~\ref{sum-reduction}, we show a reduction from any instance of bounded $k$-sum with range $\rmax$ to an instance of subset sum such that $\tmax = O_k(\rmax)$. %The reduction also generalizes to  bounded $k$-sum.

\begin{algorithm}[ht]
	\SetAlgoLined
	\KwIn{The maximum weight threshold $\rmax$, a sequence of $\opers$ operations of either (1) addition of an item $w_j$ to $\{\Aa,\Bb,\Cc\}$, (2) deletion of an item $w_j$ from $\{\Aa,\Bb,\Cc\}$, or (3) query whether there exists triple $(a, b, c) \in \Aa\times\Bb\times\Cc$ so that $a+b=c$.}
	$\textsf{cnt}\gets 0$\;
	%$\textsf{next-milestone} \gets 1$\;
	$R \gets \emptyset$\;
%	$\Aa^+, \Aa^-, \Bb^+, \Bb^-, \Cc^+, \Cc^- \gets \emptyset$\;
	$\Aa, \Bb, -\Cc, (\Aa+\Bb), (\Aa-\Cc), (\Bb-\Cc)\gets \emptyset$\;
	
	\For{$j \in [\opers]$}{
		\If{operation $j$ is a query}{
			Answer $\textsf{Yes}$ if $\textsf{cnt} > 0$ and $\textsf{No}$ otherwise\;
			\Continue\;
		}
		\If{$j \geq i\cdot\sqrt{\rmax}$}{
			%$ i \gets i + 1$\;
			%$\textsf{next-milestone} \gets $\;
			Apply the new modifications in pile $R$ to $\Aa$, $\Bb$, and $\Cc$\;
			Recompute $(\Aa+\Bb)$, $(\Aa-\Cc)$, and $(\Bb-\Cc)$\;
			$R \gets \emptyset$\;
			%$\Xx \gets (\Xx \cup \Xx^+) \setminus \Xx^- \quad \forall{\;\Xx \in \{\Aa, \Bb, -\Cc\}}$\;
			%$\Xx^+, \Xx^- \gets \emptyset\quad\quad\quad\quad\;\forall{\;\Xx \in \{\Aa, \Bb, -\Cc\}}$\;
			%$\Yy + \Zz \gets \Yy \oplus \Zz \quad\quad\;\;\; \forall{\;\Xx \in \{\Aa, %\Bb, -\Cc\}};$  $\quad$
		%	\tcp{ $\{ \Yy, \Zz \} \gets \{\Aa, \Bb, -\Cc\} \setminus \{ \Xx \}$}
			%$\Aa \gets \Aa\cup\Aa^+\setminus\Aa^-$\; 
			%$\Bb \gets \Bb\cup\Bb^+\setminus\Bb^-$\; 
			%$-\Cc \gets  (-\Cc)\cup(-\Cc)^+\setminus(-\Cc)^-$\;
			%$(\Aa+\Bb), (\Aa-\Cc), (\Bb-\Cc) \gets \Aa \oplus \Bb, \Aa \oplus (-\Cc), \Bb \oplus (-\Cc) $\;
		}
		\If{operation $j$ is an addition to $\mathsf{X}$}{
			%  $\textsf{\emph{and}}$ $w_j \not\in \Xx^+ \textsf{ \emph{and} } (\Xx_{w_j} = 0 \textsf{\emph{ or }} w_j \in \Xx^-)$ }{
				$R \gets R \cup (+,\Xx,w_j)$\;
				% $\Xx^+ \gets \Xx^+ \cup \{w_j\}$\;
				% \Xx^- \gets \Xx^- \setminus \{w_j\}$\;
				$\textsf{cnt} \gets \textsf{cnt} + \textsf{CountTriples}(w_j, \Yy+\Zz, \Yy, \Zz,  R);\quad$ \tcp{ $\{ \Yy, \Zz \} \gets \{\Aa, \Bb, -\Cc\} \setminus \{ \Xx \}$}
		}
		\If{operation $j$ is a deletion from $\mathsf{X}$}{ 
			%\textsf{\emph{ and }} w_j \not\in \Xx^- \textsf{\emph{ and }} ( \Xx_{w_j} = 1 \textsf{\emph{ or }} w_j \in \Xx^+)$}{
				$R \gets R \cup (-,\Xx,w_j)$\;
				%$\Xx^- \gets \Xx^- \cup \{w_j\}$\; 
				% \Xx^+ \gets \Xx^+ \setminus \{w_j\}$\;
				$\textsf{cnt} \gets \textsf{cnt} -\textsf{CountTriples}(w_j, \Yy+\Zz, \Yy, \Zz, R)$\; 
		}
	}		
	%\KwResult{Write here the result }
	%initialization\;
	%\While{While condition}{
	%	instructions\;
	%	\eIf{condition}{
	%		instructions1\;
	%		instructions2\;
	%	}{
	%		instructions3\;
	%	}
	%}
	\caption{Fully Dynamic Bounded $3$-Sum\\ $O(\min{\{\opers, \rmax^{0.5} \log \rmax\}})$ amortized processing time}\label{alg:3sum}
\end{algorithm}

We start by presenting a fully dynamic algorithm for bounded $3$-sum with sublinear amortized processing time. %The main results of this section are summarized in Theorem~\ref{thm:3sum} and Corollary~\ref{thm:3sum-addition}.
In our algorithm, we keep an auxiliary dataset which we refresh after every $\sqrt{\rmax}$ operations. We call such an event a \textit{milestone}. We also maintain a list $R$ which is a pile of operations which have come after the most recent computation of the auxiliary dataset. Therefore the size of $R$ never exceeds $\sqrt{\rmax}$. 
%In Algorithm~\ref{alg:3sum}, we split $R$ into six sets $\Aa^+, \Aa^-, \Bb^+, \Bb^-, \Cc^+, \Cc^-$ which keep track of additions and deletions to each set separately. 
We also maintain a number $\textsf{cnt}$ throughout the algorithm such that at every point in time, $\textsf{cnt}$ counts the number of tuples $(a,b,c)$ such that $a \in \Aa, b \in \Bb, c \in \Cc$ and $a+b = c$. The ultimate goal is to update the value of $\textsf{cnt}$ after each operation so that we can answer query operations based on whether $\mathsf{cnt} > 0$ or not.
 %We maintain sets $\big\{ \Aa^+, \Aa^-, \Bb^+, \Bb^-, (-\Cc)^+, (-\Cc)^- \big\}$ containing new modifications since the last milestone. Here, $\Xx^+$ shows the new additions to set $\Xx$ and $\Xx^-$ shows the new deletions from set $\Xx$ for all $\Xx \in \big\{\Aa, \Bb, (-\Cc)\big\}$. We make sure that $\Xx^+ \cap \Xx^-$ is always empty.
%Whenever the total size of new modifications surpasses $\opers/\milestones$, we solve the problem from scratch and we call such checkpoints \emph{milestones} and denote them by $m_1, m_2, \ldots, m_\milestones$. To do so, we compute the following sets, where 
%$-\Cc$ contains an item with weight $-w$ for each item in $\Cc$ with weight $w$
In the beginning, and also after every $\sqrt{\rmax}$ operations, we compute the auxiliary dataset. More precisely, we compute 
\[\big\{\Aa, \Bb, -\Cc, \Aa+\Bb, \Aa-\Cc, \Bb - \Cc\big\}. \]
where $\Aa+\Bb$ (or $\Aa-\Bb$) is a set containing $a+b$ (or $a-b$) for every pair of elements $(a,b) \in \Aa \times\Bb$. Similarly, $-\Cc$ contains the negation of each number in $\Cc$.

Note that we treat each of the above data as a \textbf{non binary} vector over range $[-2\rmax, 2\rmax]$ where every element specifies in how many ways a specific number can be made. For instance $\Aa + \Bb$ is a vector over range $[-2\rmax,2\rmax]$ and index $i$ counts the number of pairs $(a,b)$ such that $a \in \Aa$ and $b \in \Bb$ and $a+b = i$. It follows that each of the vectors can be computed in time $O(\rmax \log \rmax)$ using polynomial multiplication. 
%When the number of elements in the sets is small, an alternative way to compute $\Aa + \Bb$ is iterating over all pairs which can be done in time $O(\opers^2)$.
Each time we compute the auxiliary dataset, we also start over with a new pile $R$ with no operations in it. Moreover, after the computation of the auxiliary dataset, we also compute $\textsf{cnt}$ using polynomial multiplication in time $O(\rmax \log \rmax)$. The only exception is the first time we make the auxiliary dataset in which case we know $\mathsf{cnt}$ is equal to 0 and we do not spend any time on computing it.

The auxiliary dataset remains intact until we recompute it from scratch but as new operations are added into $R$, we need to update $\textsf{cnt}$. To update $\textsf{cnt}$ after an insertion operation or a deletion operation, we need to compute the number of triples that affect $\textsf{cnt}$ and contain the newly added (or deleted) number. To explain the idea, let us assume that a number $w$ is added into $\Aa$. There are four types of triples $(w,b,c)$ that can potentially affect $\textsf{cnt}$: (Keep in mind that $b$ and $c$ may refer to some numbers that previously existed in the sets but were removed at some point.)

\begin{algorithm}[ht]
	\SetAlgoLined
	\KwIn{The weight $w_j$ of the item to be added/deleted, a set $\Yy+\Zz$ containing $y+z$ for every pair $(y,z) \in \Yy\times\Zz$, and a pile of new modifications $R$.}
	
		$\textsf{triples} \gets (\Yy + \Zz)_{(-w_j)} $\;
		% Let $\Yy^+$ and $\Yy^-$ be the subset of operations in $R$ so that 
		\For{every operation $(+,\Yy, y)$ in $R$}{
			$\textsf{triples} \gets \textsf{triples} + \Zz_{(-w_j-y)} - \big|\{(-,\Zz,-w_j-y) \in R\}\big| + \big|\{(+,\Zz,-w_j-y) \in R\}\big|$\;
		}		
		\For{every operation $(-,\Yy, y)$ in $R$}{
			$\textsf{triples} \gets \textsf{triples} - \Zz_{(-w_j-y)} + \big|\{(-,\Zz,-w_j-y) \in R\}\big| $\;
		}
		\For{every operation $(+,\Zz, z)$ in $R$}{
			$\textsf{triples} \gets \textsf{triples} + \Yy_{(-w_j-z)} - \big|\{(-,\Yy,-w_j-z) \in R\}\big|$\;
		}
		\For{every operation $(-,\Zz, z)$ in $R$}{
			$\textsf{triples} \gets \textsf{triples} - \Yy_{(-w_j-z)} $\;
		}
		%$\textsf{triples} \gets \textsf{triples}  + \sum_{y \in \Yy^+}{\Zz_{(-w_j-y)}}
		%- \sum_{z \in \Zz^-}{\Yy_{(-w_j-z)}};\quad\quad$ \tcp{$(y, z) \in ((\Yy^+\cup\Yy^-)\times\Zz)$}
		%$\textsf{triples} \gets \textsf{triples}  + \sum_{z \in \Zz^+}{\Yy_{(-w_j-z)}}
		%- \sum_{z \in \Zz^-}{\Yy_{(-w_j-z)}};\quad\quad$ \tcp{$(y, z) \in  (\Yy\times(\Zz^+\cup \Zz^-))$}
		%$\textsf{triples} \gets \textsf{triples} + \big|\{y \in \Yy^+ \mid (-w_j-y) \in \Zz^+\}\big|;
		%\quad\quad\quad\quad\;\;\;$\
		%\tcp{$(y, z) \in \Yy^+\times\Zz^+$}
		$\textsf{return } \textsf{ triples}$\;
	%\KwResult{Write here the result }
	%initialization\;
	%\While{While condition}{
	%	instructions\;
	%	\eIf{condition}{
	%		instructions1\;
	%		instructions2\;
	%	}{
	%		instructions3\;
	%	}
	%}
	\caption{$\textsf{CountTriples}(w_j, \Yy+\Zz, \Yy, \Zz, R)$ \\ Counting the number of triples $(w_j,y,z) \in \{w_j\}\times\Yy\times\Zz$ so that $w_j + y + z = 0$.}\label{alg:triples}
\end{algorithm}

\begin{itemize}
	\item Both $b$ and $c$ have arrived before the last milestone and thus they are incorporated in it. The number of such pairs is $\big(\Bb-\Cc\big)_{-w}$.
	\item $b$ is added or deleted after the last milestone, but $c$ has come before the last milestone. In this case, we iterate over all new modifications to $\Bb$, which their count is bounded by $|R|$. For each of them we look up $(-\Cc)_{-w-b}$ to verify if there is a triple containing $w$ and $b$. %We need to subtract the count of new modifications that delete item $b'$ while $b'$ existed in $\Bb$ during the last milestone.
	\item In this case, $b$ refers to an operation prior to the last milestone and $c$ refers to an operation after the last milestone. This case is similar to the previous case and we just iterate over new modifications to $\Cc$.
	\item In this case, both $b$ and $c$ refer to operations after the last milestone. The number of such pairs can also be computed in time $O(|R|)$ since by fixing $b$, we just search to verify if $c = w+b$ exists in the new set of operations.
\end{itemize}

Depending on whether or not an operation adds a number or removes a number, we may incorporate the number of tuples into $\textsf{cnt}$ positively or negatively. However, since all the above counts are available in time $O(|R|)$, we can update $\mathsf{cnt}$ in time $O(|R|)$.

\begin{theorem}
	There is an algorithm for fully dynamic bounded $3$-sum whose amortized processing time is bounded by  $O(\min{\{\opers, \rmax^{0.5} \log \rmax\}})$. 
	%$O(1+\rmax\cdot \min{\{\sqrt{\opers},\sqrt{\rmax}\}}/\opers)$.
	\label{thm:3sum}
\end{theorem}
\begin{proof}
We consider two cases separately. If $\opers < \sqrt{\rmax}$, then the only time we make the auxiliary dataset is in the beginning of the algorithm and since there are no operations at that point, this only takes time $O(1)$. From then on, every operation takes time $O(|R|)$ to be processed and since $|R| \leq \opers$, then the processing time would be bounded by $O(\sqrt{\rmax})$.

If $\opers \geq \sqrt{\rmax}$ then we compute the auxiliary dataset $O(\opers / \sqrt{\rmax})$ times and since each computation takes time $O(\rmax \log \rmax)$, the overall runtime would be bounded by $O(\opers \sqrt{\rmax} \log \rmax)$. Moreover, processing each operation takes time $O(|R|) = O(\sqrt{\rmax})$. Thus, per operation the processing time is bounded by $O(\sqrt{\rmax} \log \rmax)$.
\end{proof}

%% file: src/ksum.tex
\subsection{Extending the Algorithm to Bounded $k$-sum}
The fully dynamic bounded $k$-sum problem is a generalization of fully dynamic bounded $3$-sum in which instead of three sets we have $k$ sets denoted by $\mathcal{A} = \{\Aa_1, \Aa_2, \ldots, -\Aa_k\}$. Note that we use negative values for the numbers in $\Aa_k$ since in the fully dynamic bounded $k$-sum problem we aim to find a tuple $(a_1, a_2, \ldots, a_k) \in \prod_{i\in[k]}{\Aa_i}$  such that $a_1+\ldots+a_{k-1}-a_k = 0$\footnote{Recall that we defined an analogous  condition for triples in bounded $3$-sum using a slightly different format $a+b=c$, which is equivalent to $a+b-c=0$. Thus, negating the value of each number in $\Aa_k$ resolves the asymmetries between $\Aa_k$ and other $\Aa_i$'s.}. We denote by $\mathcal{A}_S$ the subset of $\mathcal{A}$ restricted to a subset $S \subseteq [k]$ of indices from $1$ to $k$. In particular, $\mathcal{A}_{\{k\}} = -\Aa_k$ where $-\Aa_k$ contains the negative value $-w$ for each number of value $w \in \Aa_k \;(w \in [0, \rmax])$. Again we are given a threshold $\rmax$ on the value of numbers, and the operations are similar to fully dynamic bounded $3$-sum. In each operation we either add or delete a number with value in range $[0, \rmax]$ to/from one of $\Aa_i$'s, or answer a query of  the following format:

\begin{itemize}
	\item Does a tuple $(a_1, a_2, \ldots, a_k) \in \prod_{i \in [k]}{\mathcal{\Aa}_{i}}$ exist such that $a_1+\ldots+a_{k-1}-a_k = 0$?
\end{itemize}

Our algorithm for fully dynamic bounded $k$-sum is similar to our algorithm for fully dynamic at a high-level. However, instead of keeping track of $6$ sequences in the auxiliary dataset, we keep track of $2^k-2$ sequences. Moreover, each sequence in the auxiliary dataset is ranged over $[-k \rmax, k\rmax]$. Also, we refresh the auxiliary dataset every $\rmax^{1/(k-1)}$ operations.

In what follows, we study the dependence of the amortized processing time of fully dynamic bounded $k$-sum on $k$. We utilize the fully dynamic algorithm for bounded $3$-sum and extend it to compute $O(2^{k})$ sets $\Aa_S$ for each subset $S \subseteq [k]$ -- corresponding to each subset of $\mathcal{A}$ --  from scratch at each milestone. For a subset $S = \{i_1, i_2, \ldots, i_{|S|}\} \subseteq [k]$, $\Aa_S$ is defined as the set of the total values $\textsf{W}(\vec{a})$ of every tuple $\vec{a} = (a_{i_1}, a_{i_2}, \ldots, a_{i_{|S|}}) \in \prod_{z \in S}{\mathcal{A}_{\{z\}}}$, i.e.,
\[\Aa_S  = \Big\{\textsf{W}(\vec{a}) = \sum_{z \in S}{a_z} \mid \vec{a} = (a_{i_1},a_{i_2},\ldots,a_{i_{|S|}}) \in \prod_{z\in S}{\mathcal{A}_{\{z\}}}\land \textsf{W}(\vec{a}) \in [-\rmax, \rmax]\Big\}.\]

It is straightforward to compute $\Aa_S$ for every $S \subset [k]$ in $\widetilde{O}\big(2^k\cdot k^2\cdot\rmax\big)$ time using polynomial multiplication.
%, or alternatively the naive algorithm if $\hat\opers \leq \rmax^{1/(k-1)}$. 
%Note that similar to fully dynamic bounded $3$-sum, we do not need to compute $\Aa_{[k]}$, and therefore the time complexity of the naive algorithm is bounded by $O(\hat\opers^{k-1})$. 
 Given $\Aa_S$ for each subset $S \subset [k]$ and the pile of new modifications $R$ since the last milestone,
 %, $\{ \Xx^+, \Xx^- \mid \Xx \in \{ \Aa_1, \ldots, \Aa_{k-1}, (-\Aa_k) \}$,
 we update $\textsf{cnt}$ after each operation by counting the number of valid tuples $(a_1, a_2, \ldots, a_k) \in \prod_{z\in [k]}{\Aa_z}$ so that $a_1+\ldots+a_{k-1}- a_k = 0$, and $a_i = w_j$ if the current operation adds (or deletes) a number of value $w_j$ to (from) $\Aa_i$.

We iterate over each $O(2^{k-1})$ types of tuples after each modification. Each $a_z \;(\forall{\; z \in [k]  \setminus \{i\}})$  either existed before the last milestone and is not removed,
%, i.e., currently  $a_z \in \Aa_z\setminus\Aa^-_z$,
or is added to $\Aa_z$ since the last milestone.
%, i.e., $a_z \in \Aa_z^+$. 
We specify each type of tuples by a subset $S' \subseteq [k] \setminus\{i\}$. 
For the special case of $S' = [k] \setminus \{i\}$, in which we do not use the new numbers since the last milestone, we simply lookup $-w_j$, or $w_j$ if $i = k$, in a precomputed set $\Aa_{([k] \setminus \{i\})}$.
For any other type, specified by a subset $S' \subset [k] \setminus \{i\}$, we lookup the number of the tuples in $\Aa_{S'}$ for all 
%$O\big((\min{\{\rmax, \hat\opers\}})^{k-2}\big)$ 
$O\big(|R|^{k-2}\big)$ 
possible combinations of newly added numbers from at most $k-2$ sets.
% $\big\{ \Aa^+_z \mid z \in [k] \setminus (S' \cup \{i\}) \big\}$.
Depending on whether the current operation is addition or deletion, we add or subtract the number of valid tuples to/from $\textsf{cnt}$.
If $\opers < \rmax^{1/(k-1)}$, we never create the auxiliary dataset and the running time per operation is bounded by $O\big(|R|^{k-2}\big) = O(\opers^{k-2})$.
Otherwise, the total running time of the algorithm per operation is equal to $\widetilde O_k(\rmax^{(k-2)/(k-1)})$.

\begin{theorem}	\label{thm:ksum}
	There is an algorithm for fully dynamic bounded $k$-sum whose amortized processing time is bounded by  $\widetilde O_k(\min{\{\opers^{k-2}, \rmax^{(k-2)/(k-1)}\}})$. 
	%$O(1+\rmax\cdot \min{\{\sqrt{\opers},\sqrt{\rmax}\}}/\opers)$.
\end{theorem}
\begin{comment}
\begin{proof}
	Similar to the range of $\hat\opers\leq \rmax^{0.5}$ in fully dynamic bounded $3$-sum,
	for the range of $\hat\opers\leq\rmax^{1/(k-1)}$ in fully dynamic bounded $k$-sum we need to recompute $O(2^k)$ sets $\Aa_S$ for every $S \subset [k]$ using the naive algorithm with the total running time of $O(2^k\cdot \hat\opers^{k-1})$. Using an identical argument to bounded $3$-sum, we find out that $\milestones^*(\hat\opers) = 1$ for this range, and therefore the amortized processing time is bounded by $O(2^k\cdot \opers^{k-2})$.
	
	When $\rmax^{1/(k-1)} < \hat\opers$, we use Fast Fourier Transform to compute each $\Aa_S$ in time $\widetilde{O}(k\cdot\rmax)$. Hence, we spend an amortized processing time of $\widetilde{O}(\milestones\cdot2^k\cdot k \cdot \rmax/\hat\opers+(\hat\opers/\milestones)^{k-2})$. The optimum number of milestones $\milestones^*(\hat\opers)$ is equal to $\hat\opers/2(k\cdot\rmax)^{1/(k-1)}$. Hence, the amortized processing time is equal to:
	\[\widetilde{O}\Big(\milestones^*(\hat\opers)\cdot2^k\cdot k \cdot \rmax/\hat\opers+\big(\hat\opers/\milestones^*(\hat\opers)\big)^{k-2}\Big) = \widetilde{O}\Big(2^k\cdot(k\cdot \rmax)^{(k-2)/(k-1)}\Big).\]
\end{proof}
\end{comment}

%We also extend our solution to the bounded $k$-sum problem, in which there are $k$ sets $\Aa_1, \Aa_2, \ldots, \Aa_k$ and we want to find a tuple $(a_1, a_2, \ldots, a_k)$ so that $a_1+a_2+\ldots+a_{k-1}=a_k$, at the cost of an increase in total running time by a factor of at least $O(2^{k-2})$. 

%% file: src/3sum-hard.tex
\subsection{A Reduction from Bounded $k$-sum to Subset Sum}

%There are a number of implications as a result of this reduction:
%\begin{itemize}
%	\item 
%\end{itemize}

\newcommand{\instsub}{\mathcal{I}}
\newcommand{\instksum}{\mathcal{K}}

Any instance $\instksum$ of (fully) dynamic bounded $k$-sum with weight threshold $\rmax$ can be converted to an equivalent (fully) dynamic subset sum  instance $\instsub$ in a way that the answer of each query is unchanged and we have target threshold $\tmax = O(k^{k+2}\cdot\rmax)$. To create the subset sum instance $\instsub$, we process the operations of $\instksum$ one by one. We add an item with weight $f(k,i)\cdot\rmax + w$ to $\instsub$ upon the arrival of a new number in $\Aa_i$ with value $w$ if $i < k$ and $f(k,k)\cdot\rmax - w$ if $i=k$. In the rest of this section, 
%Lemma~\ref{sum-reduction}, 
we construct function $f(k,i)$ and show that $\sum_{i\in[k]}f(k,i) = O(k^{k+2})$. Similarly, we remove the respective item if we have a deletion operation from some $\Aa_i$. We further show that if we query the weight $f^{\mathcal{Q}}(k) = \sum_{i}{f(k,i)}\cdot\rmax$ in the subset sum instance $\instsub$ instead of each bounded $k$-sum query in $\instksum$, we end up with the same response. We carefully choose the value of $f(k,i)$ for every $i \in [k]$ so that any possible subset that adds up to $f^{\mathcal{Q}}(k)$ includes exactly one item from each $\Aa_i$. Thus, any bounded $k$-sum instance is reduced to an equivalent subset sum instance.

\begin{theorem}\label{sum-reduction}
	Given a (fully) dynamic bounded $k$-sum instance with weight threshold $\rmax$, there is a (fully) dynamic subset sum instance with the same number of items and target threshold $\tmax \leq O(k^{k+2}\cdot\rmax)$ that preserves the query results.
	%\begin{itemize}
	%	\item For addition instead add an element with size $f(k, i) + w$ for any $w \in \mathcal{A}_{\{i\}}$.
	%	\item For a bounded $k$-sum query, query target $\sum_{i}{f(k,i)}$.
	%\end{itemize} 
	%Given a dynamic $3$-sum instance, if we create a subset sum instance which instead adds an element with size $2t+a$ for any $a \in A$, adds an element with $3t+b$ for any $b \in B$, and queries $5t + x$ for any original query $x \in [t]$, we have a reduction that preserves query results.
\end{theorem}

\begin{proof}
	It is true that if the bounded $k$-sum query result is positive, we correctly identify it since 
	\[\sum_{i\in[k-1]}{\big(f(k,i)\cdot\rmax+a_i}\big)+ f(k,k)\cdot\rmax-a_k= (a_1+\ldots+a_{k-1}-a_k) + f^{\mathcal{Q}}(k).\]
	By querying the target $f^{\mathcal{Q}}(k) $ in the subset sum instance $\instsub$, we find out whether there is a tuple $(a_1,a_2,\ldots,a_k) \in \prod_{i\in[k]}{\Aa_i}$ so that $a_1+\ldots+a_{k-1}-a_k = 0$. We also need to show that no subset of items in $\instsub$ adds up to $f^{\mathcal{Q}}(k) $ unless it contains exactly one item from every $\Aa_i$. To achieve this, we set $f(k,i) = c + k^i$ for a large enough value $c \geq k^{k+1}$. We observe that any subset containing $k+1$ items has a weight of at least 
	\[c(k+1)\rmax = (kc + c)\rmax \geq (kc+ k^{k+1})\rmax > \big(\sum_{i\in[k]}{(c + k^i)}\big)\rmax= \sum_{i\in[k]}{f(k,i)}\rmax =  f^{\mathcal{Q}}(k).\]
	
	Additionally, any subset containing at most $k-1$ items has a weight of at most 
	\[(k-1)(c+k^k+1)\rmax \leq (k-1)(c+c/k+1)\rmax = (kc-c/k+k-1)\rmax < kc\rmax < f^{\mathcal{Q}}(k).\]
	
	where we used the inequality $k^k \leq c/k$ according to the definition of $c$. Thus, any subset in $\instsub$ adding up to $f^{\mathcal{Q}}(k)$ has exactly $k$ items. Now assume that this subset contains two items from set $\Aa_k$. The total weight of this subset must be at least $(ck+2k^k)\rmax$. However,
	
	\[f^{\mathcal{Q}}(k) = \sum_{i}f(k,i)\cdot \rmax = (ck + k^k + \sum_{i < k}{k^i})\rmax < (ck+2k^k)\rmax. \]
	
	Hence, if there is a subset of $k$ items in $\instsub$ that contains at least two items from $\Aa_k$, its total weight is strictly larger than $f^{\mathcal{Q}}(k)$. Using the same argument we can show any subset of $\instsub$ that adds up to $f^{\mathcal{Q}}(k)$ contains exactly one item from each $\Aa_i$.
	%We observe that 
	%since every added item has a weight $w \in [\rmax]$, 
	%a subset in $\instsub''$ adding up to $t = \big(2+(k-1)\cdot 2^k\big)\cdot\rmax$ has at most $2k$ items since the minimum item has a weight of at least $f'(k,1)\cdot\rmax = (2^k-2^{k-1})\cdot\rmax = 2^{k-1}\cdot\rmax$. Thus, we set $f(k,i) = 3k\cdot f'(k,1)$ for all $1 \leq i \leq k$ to prevent the overflow of weights in a subset that adds up to $\big(1+\sum_{i\in[k]}{f(k,i)}\big)\cdot\rmax$. Recall that we put an item with weight $f(k,i)\cdot\rmax+w$ in the subset sum instance $\instsub$ for any item with weight $w$ in $\Aa_i$. In this instance, the total  weight of any subset of at most $2k$ items  modulo $3k\cdot\rmax$ is bounded by $2k\cdot\rmax$, and the summation of original weights never violates the $3k\cdot\rmax$ threshold. 
	%Since we only changed $f'(k,i)$ by scaling to achieve $f(k,i)$, all the previously mentioned properties of $f'(k,i)$ holds for $f(k,i)$. In particular, any subset adding up to  $\big(1+\sum_{i\in[k]}{f(k,i)}\big)\cdot\rmax$ contains exactly one item from each original set $\Aa_i$ in the bounded $k$-sum instance $\instksum$. \Hamed{fix}
	This completes the reduction from any (fully) dynamic bounded $k$-sum instance with weight threshold $\rmax$ to a (fully) dynamic subset sum instance with the following target threshold if we set $c = k^{k+1}$.
	\[\tmax = f^{\mathcal{Q}}(k) = \sum_i{f(k,i)\cdot\rmax} = (kc+\sum_{i}{k^i})\cdot\rmax < (k^{k+2}+k^{k+1})\rmax = O(k^{k+2}\cdot \rmax).\]
\end{proof}

%% file: subsetsum.bbl
\begin{thebibliography}{10}

\bibitem{abboud2019seth}
A.~Abboud, K.~Bringmann, D.~Hermelin, and D.~Shabtay.
\newblock Seth-based lower bounds for subset sum and bicriteria path.
\newblock In {\em SODA 2019}.

\bibitem{DBLP:conf/soda/AssadiOSS19}
S.~Assadi, K.~Onak, B.~Schieber, and S.~Solomon.
\newblock Fully dynamic maximal independent set with sublinear in n update
  time.
\newblock In {\em SODA 2019}.

\bibitem{DBLP:conf/stoc/AssadiOSS18}
S.~Assadi, K.~Onak, B.~Schieber, and S.~Solomon.
\newblock Fully dynamic maximal independent set with sublinear update time.
\newblock In {\em STOC 2018}.

\bibitem{axiotis2021fast}
K.~Axiotis, A.~Backurs, K.~Bringmann, C.~Jin, V.~Nakos, C.~Tzamos, and H.~Wu.
\newblock Fast and simple modular subset sum.
\newblock In {\em SOSA 2021}.

\bibitem{axiotis2019fast}
K.~Axiotis, A.~Backurs, C.~Jin, C.~Tzamos, and H.~Wu.
\newblock Fast modular subset sum using linear sketching.
\newblock In {\em SODA 2019}.

\bibitem{bateni2018fast}
M.~Bateni, M.~Hajiaghayi, S.~Seddighin, and C.~Stein.
\newblock Fast algorithms for knapsack via convolution and prediction.
\newblock In {\em STOC 2018}.

\bibitem{DBLP:journals/corr/abs-1909-03478}
S.~Behnezhad, M.~Derakhshan, M.~Hajiaghayi, C.~Stein, and M.~Sudan.
\newblock Fully dynamic maximal independent set with polylogarithmic update
  time.
\newblock In {\em FOCS 2019}.

\bibitem{behrend1946sets}
F.~A. Behrend.
\newblock On sets of integers which contain no three terms in arithmetical
  progression.
\newblock {\em Proceedings of the National Academy of Sciences of the United
  States of America}, 32(12):331, 1946.

\bibitem{bellman57a}
R.~E. Bellman.
\newblock {\em Dynamic Programming}.
\newblock Courier Dover Publications, 1957.

\bibitem{bringmann2017near}
K.~Bringmann.
\newblock A near-linear pseudopolynomial time algorithm for subset sum.
\newblock In {\em SODA 2017}.

\bibitem{calabro2009complexity}
C.~Calabro, R.~Impagliazzo, and R.~Paturi.
\newblock The complexity of satisfiability of small depth circuits.
\newblock In {\em IPEC 2009}.

\bibitem{cardinal2021modular}
J.~Cardinal and J.~Iacono.
\newblock Modular subset sum, dynamic strings, and zero-sum sets.
\newblock In {\em SOSA 2021}.

\bibitem{charalampopoulos2020dynamic}
P.~Charalampopoulos, T.~Kociumaka, and S.~Mozes.
\newblock Dynamic string alignment.
\newblock In {\em CPM 2020}.

\bibitem{chen2013dynamic}
A.~Chen, T.~Chu, and N.~Pinsker.
\newblock The dynamic longest increasing subsequence problem.
\newblock {\em arXiv preprint arXiv:1309.7724}, 2013.

\bibitem{gawrychowski2020fully}
P.~Gawrychowski and W.~Janczewski.
\newblock Fully dynamic approximation of lis in polylogarithmic time.
\newblock In {\em STOC 2021}.

\bibitem{gawrychowski2018optimal}
P.~Gawrychowski, A.~Karczmarz, T.~Kociumaka, J.~{\L}acki, and P.~Sankowski.
\newblock Optimal dynamic strings.
\newblock In {\em SODA 2018}.

\bibitem{DBLP:conf/stoc/HenzingerKNS15}
M.~Henzinger, S.~Krinninger, D.~Nanongkai, and T.~Saranurak.
\newblock Unifying and strengthening hardness for dynamic problems via the
  online matrix-vector multiplication conjecture.
\newblock In {\em STOC 2015}.

\bibitem{horowitz1974computing}
E.~Horowitz and S.~Sahni.
\newblock Computing partitions with applications to the knapsack problem.
\newblock {\em JACM}.

\bibitem{impagliazzo2001complexity}
R.~Impagliazzo and R.~Paturi.
\newblock On the complexity of $k$-sat.
\newblock {\em Journal of Computer and System Sciences}.

\bibitem{impagliazzo2001problems}
R.~Impagliazzo, R.~Paturi, and F.~Zane.
\newblock Which problems have strongly exponential complexity?
\newblock {\em Journal of Computer and System Sciences}.

\bibitem{saeednew}
T.~Kociumaka and S.~Seddighin.
\newblock Improved dynamic algorithms for longest increasing subsequence.
\newblock In {\em STOC 2021}.

\bibitem{koiliaris2019faster}
K.~Koiliaris and C.~Xu.
\newblock Faster pseudopolynomial time algorithms for subset sum.
\newblock In {\em SODA 2017}.

\bibitem{DBLP:conf/stoc/LackiOPSZ15}
J.~Lacki, J.~Ocwieja, M.~Pilipczuk, P.~Sankowski, and A.~Zych.
\newblock The power of dynamic distance oracles: Efficient dynamic algorithms
  for the steiner tree.
\newblock In {\em STOC 2015}.

\bibitem{our-stoc-paper}
M.~Mitzenmacher and S.~Seddighin.
\newblock Dynamic algorithms for \textsf{LIS} and distance to monotonicity.
\newblock In {\em STOC 2020}.

\bibitem{DBLP:conf/stoc/NanongkaiS17}
D.~Nanongkai and T.~Saranurak.
\newblock Dynamic spanning forest with worst-case update time: adaptive, las
  vegas, and o(n\({}^{\mbox{1/2 - {\(\epsilon\)}}}\))-time.
\newblock In {\em STOC 2017}.

\bibitem{DBLP:conf/focs/NanongkaiSW17}
D.~Nanongkai, T.~Saranurak, and C.~Wulff{-}Nilsen.
\newblock Dynamic minimum spanning forest with subpolynomial worst-case update
  time.
\newblock In {\em FOCS 2017}.

\end{thebibliography}
